\documentclass[lettersize,journal]{IEEEtran}
\usepackage{amsfonts}
\usepackage{amssymb}
\usepackage{amsmath}
\usepackage{algorithmic}
\usepackage{array}
\usepackage{soul}
\usepackage{multirow}
\usepackage[table]{xcolor}
\PassOptionsToPackage{hyphens}{url}
\usepackage{hyperref}

\usepackage{booktabs} 
\usepackage{textcomp}
\usepackage[ruled]{algorithm2e}
\usepackage{stfloats}
\usepackage{dsfont}
\usepackage{url}
\usepackage{xcolor}
\usepackage{subcaption}
\usepackage{verbatim}
\usepackage{graphicx}
\allowdisplaybreaks[4]
\usepackage{cite}
\usepackage{upgreek}
\usepackage{pifont}
\usepackage{makecell}
\usepackage{enumitem}
\newtheorem{theorem}{Theorem}
\newtheorem{definition}{Definition}
\newtheorem{remark}{Remark}

\newtheorem{lemma}{Lemma}

\ifodd 1
\newcommand{\rev}[1]{{\color{blue}#1}} 
\else
\newcommand{\rev}[1]{#1}
\fi

\begin{document}

\title{Birdcast: Interest-aware BEV Multicasting for Infrastructure-assisted Collaborative Perception}


\author{Yanan Ma, Zhengru Fang, Yihang Tao, Yu Guo, Yiqin Deng,~\IEEEmembership{Member,~IEEE}, \\
Xianhao Chen,~\IEEEmembership{Member,~IEEE}, and Yuguang Fang,~\IEEEmembership{Fellow,~IEEE}
\thanks{
The work was supported in part by the JC STEM Lab of Smart City funded by The Hong Kong Jockey Club Charities Trust under Contract 2023-0108, in part by the Research Grants Council of the Hong Kong SAR, China (Project No. CityU 11216324), and in part by the Hong Kong SAR Government under the Global STEM Professorship. 
The work of Y. Deng was supported in part by the National Natural Science Foundation of China under Grant No. 62301300 and in part by the Shandong Province Science Foundation under Grant No. ZR2023QF053.
The work of X. Chen was supported in part by the Research Grants Council of Hong Kong under Grant 27213824 and CRS HKU702/24.}
\thanks{Y. Ma, Z. Fang, Y. Tao, Y. Guo, and Y. Fang are with the Hong Kong JC Lab of Smart City and the Department of Computer Science, City University of Hong Kong, Hong Kong, China. E-mail: yananma8-c@my.cityu.edu.hk, zhefang4-c@my.cityu.edu.hk, yihang.tommy@my.cityu.edu.hk, yu.guo@my.cityu.edu.hk, my.fang@cityu.edu.hk.}
\thanks{Y. Deng is with the School of Data Science, Lingnan University, Tuen Men, Hong Kong, China. E-mail: yiqindeng@ln.edu.hk.}
\thanks{X. Chen is with the Department of Electrical and Electronic Engineering, University of Hong Kong, Hong Kong, China. E-mail: xchen@eee.hku.hk.}
}



\maketitle
\begin{abstract}
Vehicle-to-infrastructure collaborative perception (V2I-CP) leverages a high-vantage node to transmit supplementary information, i.e., bird's-eye-view (BEV) feature maps, to vehicles, effectively overcoming line-of-sight limitations. However, the downlink V2I transmission introduces a significant communication bottleneck. Moreover, vehicles in V2I-CP require \textit{heterogeneous yet overlapping} information tailored to their unique occlusions and locations, rendering standard unicast/broadcast protocols inefficient. To address this limitation, we propose \textit{Birdcast}, a novel multicasting framework for V2I-CP. By accounting for individual maps of interest, we formulate a joint feature selection and multicast grouping problem to maximize network-wide utility under communication constraints. Since this formulation is a mixed-integer nonlinear program and is NP-hard, we develop an accelerated greedy algorithm with a theoretical $(1 - 1/\sqrt{e})$ approximation guarantee. While motivated by CP, Birdcast provides a general framework applicable to a wide range of multicasting systems where users possess heterogeneous interests and varying channel conditions. 
Extensive simulations on the V2X-Sim dataset demonstrate that Birdcast significantly outperforms state-of-the-art baselines in both system utility and perception quality, achieving up to 27\% improvement in total utility and a 3.2\% increase in mean average precision (mAP).

\end{abstract}


\section{Introduction}
\IEEEPARstart{A}{ccurate} perception serves as a fundamental pillar for safe navigation and decision-making in autonomous driving~\cite{chen2024vehicle, Sun_2020_CVPR, ma2025sense4fl, tao2025gcpguarded}. 
However, while the onboard perception capabilities of connected and autonomous vehicles (CAVs) have advanced significantly, single-vehicle perception remains inherently hindered by physical occlusions, such as heavy trucks and high-rise buildings, especially in busy roads and high-density urban environments \cite{yin2025occlusion}. Many real-world failures, such as a self-driving vehicle crashing into a railroad crossing arm~\cite{Tesla}, highlight the critical risks associated with these line-of-sight (LoS) blind spots. To overcome these LoS barriers, collaborative perception (CP) via vehicle-to-everything (V2X) communication has emerged as a transformative solution.
By sharing complementary sensory information among distributed nodes, CP extends the effective perception range of individual CAVs, enabling them to ``see'' through obstacles~\cite{hu2025collaborative, fang2024pacp,fang2025r}. 
To this end, vehicle-to-infrastructure  CP (V2I-CP) has emerged as a powerful paradigm~\cite{9228884}. 
In V2I-CP, a high-vantage node (HVN), such as a roadside unit (RSU) equipped with advanced cameras/LiDAR sensors, disseminates the information, i.e., the processed bird's-eye view (BEV) feature grids\footnote{Without loss of generality, we assume the spatial representations are partitioned into multiple BEV feature grids for transmissions, hereafter referred to simply as BEV grids.}, to CAVs, thereby overcoming their inherent LoS limitations~\cite{deng2025uav, shi2024soar, fang2025task}.

\begin{figure}
    \centering
    \includegraphics[width=1.0\linewidth]{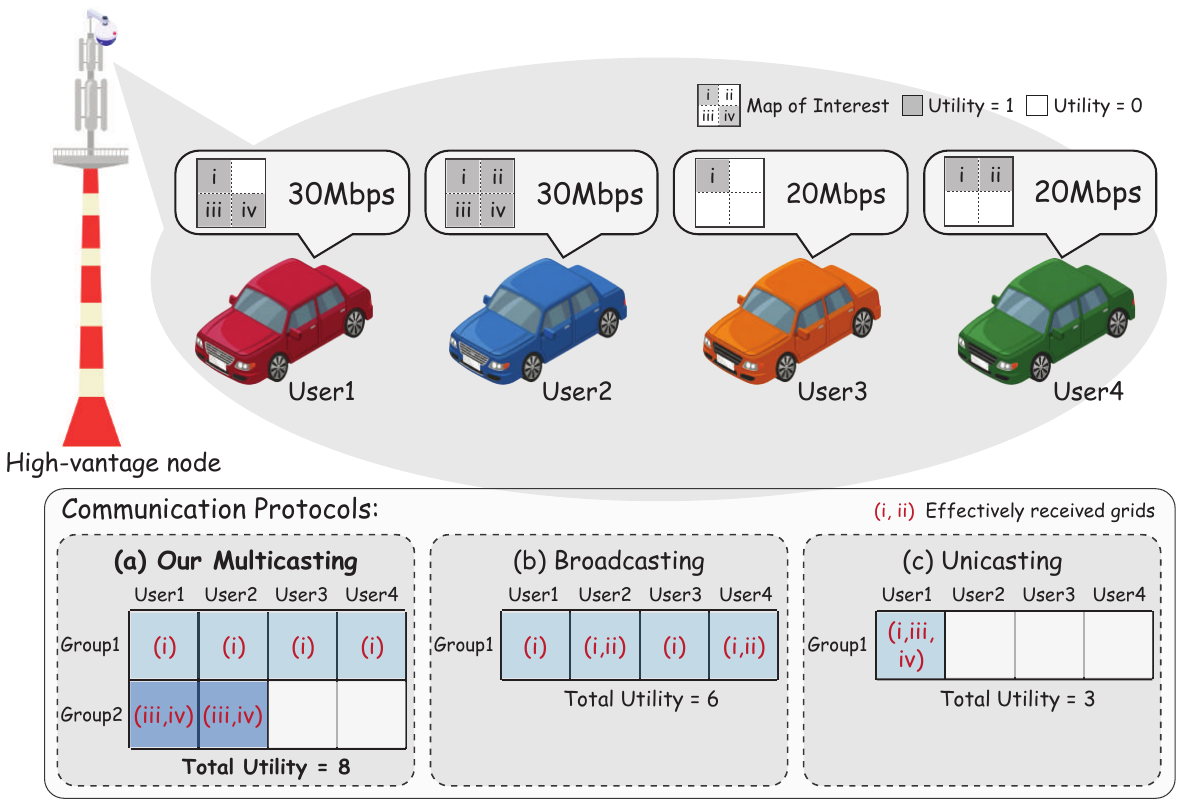}
    \caption{Limitations of broadcast and unicast in V2I-CP. A high-vantage node (HVN) transmits BEV grids to four users with heterogeneous requests (gray quadrants) and varying data rates under a latency budget of 14 ms (the data volume for one grid is 15 KB). Broadcast and unicast cannot efficiently exploit overlapping user interests, yielding suboptimal utilities of 6 and 3, respectively. In contrast, multicasting optimally groups users based on shared requests, resulting in a utility of 8.
    }
    \label{fig:toyexample}
    \vspace{-0.0cm}
\end{figure}

Nevertheless, the transmission of high-dimensional BEV grids struggles to satisfy the stringent latency requirements of autonomous driving under communication constraints~\cite{hu2022where2comm, fang2024pacp}. A core issue lies in the fact that the information requests of the users, i.e., vehicles in this system, are inherently \textit{heterogeneous yet overlapping}. Specifically, while adjacent vehicles share some common regions of interest (RoIs), they also demand different sensor data due to varied occlusions, sensor footprints, and trajectories.
However, the communication schemes underpinning V2I-CP in the existing literature are typically traditional broadcast or unicast protocols, both of which are inefficient in such scenarios. Specifically, broadcast protocols lead to severe resource underutilization by transmitting irrelevant data to users with distinct requests, whereas unicast fails to exploit the overlapping information requests among neighboring users. Furthermore, traditional multicasting schemes cannot adequately resolve this issue, as they largely assume all users require identical content (e.g., a popular video stream \cite{chen2015fair, elbadry2024wireless, zhang2021joint}), resulting in suboptimal network performance for V2I-CP.

To illustrate these limitations, Fig. \ref{fig:toyexample} presents a toy example of an HVN serving four users with varying data rates and distinct requirements, i.e., maps of interest (MoIs). Under a 14 ms latency budget, traditional broadcast is constrained by the weakest vehicle (20 Mbps), yielding a suboptimal system utility of 6. Unicast sequentially serves individual users without exploiting shared requests, resulting in the lowest utility of 3. In contrast, an optimal strategy groups users by leveraging their overlapping interests. By transmitting grid $i$ to all users and the shared grids $(iii, iv)$ to stronger users, this approach achieves the highest utility of 8.

To address these communication inefficiencies, this paper proposes \textit{Birdcast}, a novel optimization framework for V2I-CP, which introduces interest-aware BEV multicasting through jointly optimized content selection, i.e., feature selection, and multicast grouping. Concretely, we first design a protocol in which the HVN and individual users generate a customized MoI for each user, driven by its specific RoI and prediction confidence. Guided by these generated MoIs, we formulate an optimization problem to maximize the total utility under the latency budget constraint. Since the resultant problem is NP-hard, we transform it into an equivalent \textit{grid-rate} selection problem and demonstrate that the transformed problem is a monotone submodular maximization problem subject to a knapsack constraint. Finally, we develop an accelerated greedy algorithm that achieves a $(1-1/\sqrt{e})$-approximation guarantee with provably low computational complexity, making it feasible and suitable for the real-time multicasting optimization required for time-sensitive CP applications.

While motivated primarily by the challenges of V2I-CP, this optimization framework is broadly applicable to general multicasting problems where users have heterogeneous yet overlapping requests. To the best of our knowledge, Birdcast is the \textit{first optimization framework for joint content selection and multicast grouping over such overlapping requests while providing a theoretical approximation guarantee}. 
The main contributions of this work are summarized as follows:

\begin{itemize}
    \item We propose Birdcast, a novel multicast framework that utilizes an HVN to complement the perceptual views of individual users. Under this framework, we design a protocol that allows the HVN and users to collectively construct user-specific maps of interest based on their specific regions of interest and prediction confidence.

    \item Guided by these maps of interest, we formulate the joint feature selection and multicast grouping problem and rigorously prove its NP-hardness. To achieve tractability, we transform the problem from a novel \textit{grid-rate} perspective and mathematically prove that the optimal solution to this transformed problem exactly matches that of the original.

    \item By exploiting the submodularity of the transformed objective function, we develop an accelerated greedy algorithm.  We rigorously demonstrate that this algorithm achieves a provable approximation ratio of $1 - \sqrt{1/e}$ while maintaining the low computational complexity required for real-time execution. 
    
    \item We conduct extensive simulations using the V2X-Sim dataset to validate the superiority of Birdcast against baseline approaches. Our results demonstrate significant improvements, achieving up to 27\% in total system utility and a 3.2\% gain in mean average precision (mAP).

\end{itemize}

The remainder of this paper is organized as follows. Section \ref{sec:related_work} introduces related work. Section \ref{sec:framework} elaborates on the system model and the Birdcast framework. Section \ref{sec:problem} formulates the interest-aware multicasting optimization problem. The problem transformation is presented in Section \ref{sec:reformulation}. Section \ref{sec:algorithm} offers the corresponding solution approach and theoretical analysis. Section \ref{sec:simulation} provides the simulation results. Finally, concluding remarks are presented in Section \ref{sec:conclusion}.

\section{Related Work}
\label{sec:related_work}

\subsection{Collaborative Perception}
Collaborative Perception (CP) has emerged as a pivotal paradigm to extend the sensing range of vehicles and overcome line-of-sight occlusions~\cite{fang2025r, hu2025collaborative}. Based on the level of data shared, CP strategies are generally categorized into three paradigms: early fusion (sharing raw sensory data), late fusion (sharing post-detection outputs), and intermediate fusion (sharing extracted deep features). 

In recent years, feature sharing has attracted significant attention because it strikes a better tradeoff between perception accuracy and communication costs. For instance, Who2com~\cite{liu2020who2com} and When2com~\cite{liu2020when2com} introduced handshake mechanisms to learn communication graphs, selecting partners based on potential information gain. Where2comm~\cite{hu2022where2comm} introduced spatial confidence maps to selectively transmit only perceptually critical information. CoSDH~\cite{xu2025cosdh} developed a hybrid fusion scheme that dynamically switches between intermediate and late fusion based on a supply-demand awareness module. Furthermore, PACP \cite{fang2024pacp} formulated a priority-aware framework, leveraging submodular optimization to schedule communication links based on the correlations between the ego vehicle and its neighbors.

Given Vehicle-to-Vehicle (V2V) CP remains still vulnerable to severe occlusions, another line of research incorporates HVNs, e.g., RSUs, to provide an extended field of view~\cite{bae2024rethinking, 10909254}. For instance, V2IViewer~\cite{yi2024v2iviewer} proposed a heterogeneous multi-agent middle layer to efficiently align infrastructure and vehicle features for robust long-range detection.  Directed-CP~\cite{tao2025directed} introduced an RSU-aided proactive attention mechanism that guides the ego user to prioritize feature aggregation from safety-critical sectors. AirV2X~\cite{gao2025airv2x} established a unified benchmark for air-ground collaboration, providing a high-fidelity dataset that integrates LiDAR-camera data across CAVs, RSUs, and drones. Despite these advancements, these frameworks have not addressed underlying communication bottlenecks from a resource optimization perspective. Consequently, they continue to rely on naive broadcast or unicast protocols, leading to inefficient network utilization.

\begin{table}[!t]
\centering
\caption{Summary of related works in resource allocation for wireless multicasting.}
\label{table:related_work}
\resizebox{\linewidth}{!}
{
\renewcommand{\arraystretch}{1.4}
\setlength{\tabcolsep}{2mm}
\begin{tabular}{|c|c|c|c|}
\hline
{\textbf{Ref.}}& 
\makecell[c]{\textbf{Heterogeneous} \\ \textbf{Requests}}& 
\makecell[c]{\textbf{Heterogeneous} \\ \textbf{Channel Conditions}}& 
\makecell[c]{\textbf{Performance} \\ \textbf{Guarantee}}
  \\ \hline
\cite{chen2015fair} 
    & {\ding{55}}  & {\ding{52}}  & {\ding{52}}   \\ \hline
\cite{elbadry2024wireless} 
    & {\ding{55}}  & {\ding{52}}  & {\ding{55}}  \\ \hline 
\cite{zhang2021joint} 
    & {\ding{55}}  & {\ding{52}}  & {\ding{55}}  \\ \hline     
\cite{lyu2022distributed} 
    & {\ding{55}}  & {\ding{52}}  & {\ding{52}} \\ \hline
\cite{lin2012video} 
    & {\ding{52}}  & {\ding{52}}  & {\ding{55}} \\ \hline
Ours 
    & {\ding{52}}  & {\ding{52}}  & {\ding{52}}  \\ \hline
\end{tabular}
}
\end{table}

\subsection{Resource Allocation in Wireless Multicasting} 
Wireless multicasting is a fundamental technique for enhancing spectral efficiency by simultaneously serving multiple users requesting the same data~\cite{wang2019efficient, elbadry2024wireless, chen2015fair, zhang2021joint}. A primary challenge in multicast transmission is the ``straggler'' problem, where the achievable data rate of a group is constrained by the user with the worst channel conditions. To mitigate this, resource allocation is the key. Elbadry \textit{et al.}~\cite{elbadry2024wireless} proposed an adaptive multicast rate control scheme designed to satisfy the goodput and loss requirements of the majority, thereby preventing system performance from being severely constrained by stragglers. 
Zhang \textit{et al.}~\cite{zhang2021joint} developed a two-step multicast framework for live video delivery. They first devised a K-means++ clustering algorithm to automatically determine the number of multicast groups and user assignment, and then employed Lyapunov optimization to manage group-level resource allocation and version selection based on the fixed grouping.
Lyu \textit{et al.}~\cite{lyu2022distributed} developed a distributed graph-based framework to maximize throughput by jointly optimizing multicast link establishment with data dissemination decisions. Chen \textit{et al.}~\cite{chen2015fair} proposed a dynamic programming approach for joint group partitioning and resource allocation in eMBMS. 

However, the existing literature often assumes that all users in the multicasting system require identical information. In practice, in V2I-CP scenarios, a salient characteristic is that users have heterogeneous yet overlapping information interests. Very few studies have addressed the multicasting optimization under such heterogeneous user requests. Lin \textit{et al.} \cite{lin2012video} proposed a marginal-utility-based video multicast scheme that allocates bandwidth by considering both varying channel conditions and diverse user interests. Nevertheless, this approach is tailored to conventional video streams and relies on heuristics that lack theoretical performance guarantees.

To the best of our knowledge, the joint optimization of content selection and multicast grouping for heterogeneous requests remains an open challenge. To bridge this gap, we introduce an \textit{interest-aware multicast framework} with a performance guarantee. A comprehensive comparison between our proposed approach and the existing literature is summarized in Table \ref{table:related_work}.

\section{The Proposed Birdcast Framework}
\label{sec:framework}

In this section, we present the Birdcast framework for the V2I-CP system. We detail the system model, the communication model, the construction of the map of interest, and the feature aggregation process.

\begin{figure}
    \centering
    \includegraphics[width=1.0\linewidth]{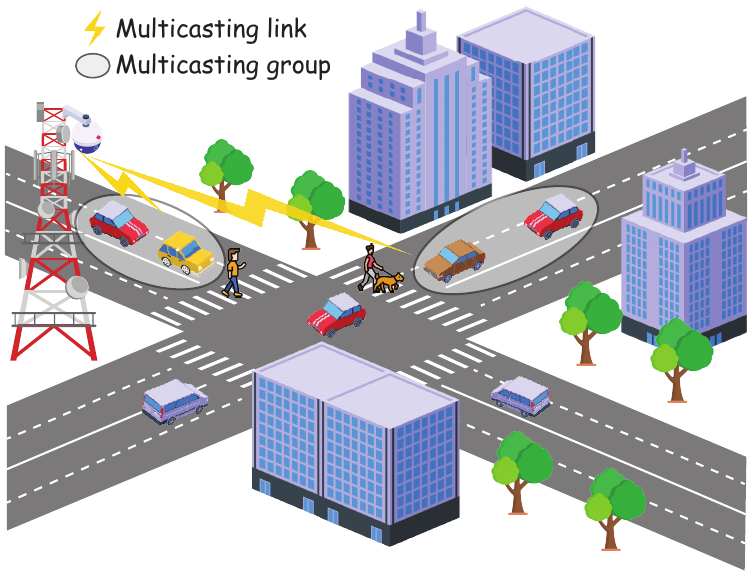}
    \caption{Illustration of the Birdcast framework. The HVN leverages its elevated altitude to capture a comprehensive BEV. To enhance perception performance and communication efficiency, the users and the HVN collaboratively identify a map of interest to assign spatial weights to each BEV grid. Consequently, the HVN can selectively transmit data to users by jointly considering the users' channel conditions and heterogeneous interests, i.e., prioritizing high-utility grids that yield substantial performance gains.}
    \label{fig:system}
    \vspace{-0.0cm}
\end{figure}

\subsection{System Model}
As illustrated in Fig. \ref{fig:system}, we consider a vehicular network comprising one HVN (e.g., an RSU) and a set of users (e.g., ground CAVs), denoted as $\mathcal{N}=\{1, \dots, N\}$. Ground users, positioned at $\mathbf{p}_n=(x_n, y_n, 0)$, frequently suffer from LoS occlusions caused by static infrastructure and dynamic traffic (e.g., trucks and buses). To mitigate this, the HVN, located at an elevated position $\mathbf{p}_{\mathrm{H}}=(x_{\mathrm{H}}, y_{\mathrm{H}}, z_{\mathrm{H}})$, can capture a comprehensive BEV of the environment via camera and/or LiDAR and then disseminate the information to users.

Broadcasting all BEV grids is highly inefficient and users may receive redundant information they do not really need, such as BEV grids with short distances and clear LoS, or grids with longer distances and in opposite directions, as shown in Fig. \ref{fig:system}. To improve communication efficiency, the HVN and users have to jointly identify a map of interest (MoI) by assigning weights to each grid. Based on this, the HVN can selectively multicast data from high-utility grids that provide significant performance gains to the corresponding users, thereby facilitating robust trajectory planning and ensuring safety in complex traffic environments~\cite{hu2022where2comm, xu2025cosdh}.

\subsection{Communication Model}
We consider downlink transmission from the HVN to users via a cellular V2I channel. To accommodate varying channel conditions, HVN dynamically selects an appropriate modulation and coding scheme (MCS) from a discrete set to reliably multicast information to users \cite{chen2015fair, tan2008link}.

We determine MCS configurations based on their corresponding link-level spectral efficiency.
Let $\mathcal{R}=\{r_1, r_2, \dots, r_M\}$ denote the finite set of MCS rate options (in bits/s/Hz), where $r_i < r_j$ for $i<j$, and $M$ is the cardinality of the set. We define a corresponding signal-to-noise ratio (SNR) threshold set $\Gamma=\{\gamma_{\mathrm{th}, 1}, \gamma_{\mathrm{th}, 2}, \dots, \gamma_{\mathrm{th}, M+1}\}$, where $\gamma_{\mathrm{th}, M+1}=\infty$. For reliable decoding, each MCS rate $r_m$ requires the received SNR to meet or exceed its lower bound $\gamma_{\mathrm{th}, m}$. 
Specifically, the set of MCS options that can be adopted for user $n$ is given by 
\begin{equation}
    \mathcal{C}_n  = \{ r_m \in \mathcal{R} \mid \gamma_n \geq \gamma_{\mathrm{th}, m}  \},
\end{equation}
where $\gamma_n$ is the received SNR for user $n$ estimated based on channel state information. 
Let $r_n^{\max} = \max\{r_j \mid r_j \in \mathcal{C}_n\}$ denote the maximum MCS rate decodable by user $n$. The corresponding maximum achievable data rate for user $n$ is defined as:
\begin{equation}
   R_n = B r_n^{\max},
\end{equation}
where $B$ denotes the spectrum bandwidth for multicasting.


\subsection{Construction of Map of Interest}
To facilitate communication-efficient CP, we construct a user-specific MoI, denoted by $\mathbf{M}_n \in [0,1]^{H \times W}$, which assigns an importance weight to each BEV grid for user $n$. To be specific, we assume the HVN and each user process raw sensor observations $\mathbf{B}_{\mathrm{HVN}}$ and $\mathbf{B}_n$ (e.g., 3D point clouds) using a shared feature encoder $\Phi_{\mathrm{Enc}}(\cdot)$ to extract high-dimensional BEV representations. The resulting feature maps are denoted as $\mathbf{F}_{\mathrm{HVN}}=\Phi_{\mathrm{Enc}}(\mathbf{B}_{\mathrm{HVN}}) \in \mathbb{R}^{H \times W \times V}$ and $\mathbf{F}_n=\Phi_{\mathrm{Enc}}(\mathbf{B}_n) \in \mathbb{R}^{H \times W \times V}$, where $H, W$, and $V$ represent the height, width, and channel dimensions, respectively. Subsequently, the HVN and the users collaboratively construct the map of interest $\mathbf{M}_n$. The interest level of each grid is jointly determined by two core factors: i) the inherent information density of the HVN's features at that location, and ii) the confidence gap between HVN and users. 

\textit{\textbf{i) Construction of Information Map $\mathbf{M}_{\mathrm{info}}$:}}
The utility of a transmitted BEV grid is intrinsically related to its information density. For instance, homogeneous areas (e.g., empty road surfaces) are perceptually sparse, and hence transmitting BEV grids from such uninformative regions consumes valuable bandwidth without yielding meaningful perceptual gains.

Motivated by the entropy-based filtering mechanism in \cite{wang2023umc}, we evaluate the information density of the HVN's feature map using a sliding window to filter out these uninformative areas. First, the original feature map $\mathbf{F}_{\text{HVN}} \in \mathbb{R}^{H \times W \times V}$ are compressed into $\Bar{\mathbf{F}}_{\text{HVN}} \in \mathbb{R}^{H \times W \times 1}$ using a trainable $1 \times 1$ convolution network \cite{wang2024drones}. Then, a $W_\mathrm{s} \times W_\mathrm{s}$ sliding window traverses this compressed feature map to compute the discrepancy between the central feature and its neighbors. By normalizing these values to $(0, 1)$ via a sigmoid function $\sigma$, the average across the window serves as the local correlation score $p_{i,j}$:

\begin{small}
\begin{equation}
    p_{i,j} = \frac{1}{W_\mathrm{s}^2} \sum_{w_1=0}^{W_\mathrm{s}-1} \sum_{w_2=0}^{W_\mathrm{s}-1}  \sigma \left(\Bar{\mathbf{F}}_{\text{HVN}}(i+w_1, j+w_2)- \Bar{\mathbf{F}}_{\text{HVN}}(i,j)\right).
\end{equation}
\end{small}

Based on this, the information density of each feature grid is calculated as:
\begin{equation}
[\mathbf{E}_{\text{HVN}}]_{i,j} = p_{i,j} \log p_{i,j}.
\end{equation}
Finally, a binary informativeness mask $\mathbf{M}_{\mathrm{info}} \in \{0,1\}^{H \times W}$ is generated via an element-wise thresholding function:
 \begin{equation}
\mathbf{M}_{\mathrm{info}} = f\left(\mathbf{E}_\text{HVN}, \eta \right),
 \end{equation}
where $f(\cdot, \eta)$ sets the top-$\eta$ percentile of entries in $\mathbf{E}_\text{HVN}$ to 1 while setting the remaining entries to 0. This operation effectively penalizes uniform, low-variation regions, such as empty backgrounds, while preserving critically informative areas.

\textit{\textbf{ii) Construction of Confidence Map $\mathbf{O}_n$:}}
Beyond information density, HVN should supplement information to users by considering their locally available sensor data. Following \cite{hu2022where2comm}, we quantify the perceptual disparity between the HVN and user $n$ using a spatial attention mask, denoted as the confidence map $\mathbf{O}_n \in \mathbb{R}^{H \times W}$. Intuitively, if a user can already reliably detect objects in a specific grid that yield a high local confidence score, it does not require supplementary data. Conversely, a low local score indicates uncertainty, often due to severe occlusions, necessitating assistance from the HVN.
Accordingly, the confidence map is formulated as
\begin{equation}
[\mathbf{O}_n]_{i,j} = \max\{[\mathbf{Q}_{\mathrm{HVN}}]_{i,j} - [\mathbf{Q}_n]_{i,j}, 0\},
\end{equation}
where $\mathbf{Q}_{\mathrm{HVN}}$ and $\mathbf{Q}_n$ represent the confidence scores obtained from the respective detection heads of the HVN and user $n$. This formulation explicitly prioritizes grids where the HVN possesses high perceptual confidence while the user remains uncertain.

\textbf{MoI Construction.}
Finally, the MoI for user $n$ is obtained by combining the information map, the confidence map, and the user's region of interest:
\begin{equation}
\mathbf{M}_n = \mathbf{O}_n \odot \mathbf{M}_{\mathrm{info}} \odot \mathbf{RoI}_n,
\end{equation}
where $\odot$ denotes the Hadamard (element-wise) product, and $\mathbf{RoI}_n \in \{0, 1\}^{H \times W}$ is a binary mask representing the predefined region of interest (RoI) for user $n$.  
Under this unified formulation, the resulting weight $[\mathbf{M}_n]_{i,j}$ equals $0$ if a grid lies outside the user's RoI, is uninformative ($[\mathbf{M}_{\mathrm{info}}]_{i,j} = 0$), or requires no perceptual assistance ($[\mathbf{O}_n]_{i,j} = 0$). Conversely, high values in $\mathbf{M}_n$ indicate the critical regions that the HVN should prioritize for feature multicasting.

\begin{remark}
    While we construct $\mathbf{M}_n$ using confidence scores following \cite{hu2022where2comm}, the proposed Birdcast framework in Section \ref{sec:problem} treats $\mathbf{M}_n$ as a general matrix, allowing the integration of any alternative user-interest indicators.
\end{remark}

\begin{table}[!t]
\centering
\caption{Summary of important notations.}
\label{table:notations}
\renewcommand{\arraystretch}{1.4}
\setlength{\tabcolsep}{2mm}
\begin{tabular}{@{}p{0.85cm}p{7.0cm}@{}}
\hline
\textbf{Notation} & \textbf{Description} \\ 
\hline
~$\mathcal{N}$ & The set of users\\
~$\mathcal{R}$ & The set of MCS rates\\
~$\Gamma$ & The set of SNR thresholds \\
~$r_m$ & The MCS rate (bits/s/Hz) \\
~$B$ & The transmission bandwidth \\
~$\gamma_{\mathrm{th},m}$ & The SNR lower bound to decode at MCS rate $r_m$\\
~$\gamma_n$ & The received SNR for user $n$ \\
~$\mathbf{F}_{\mathrm{HVN}}$ & The feature map of HVN\\
~$\mathbf{M}_n$ & The map of interest (MoI) for user $n$ \\
~$K$ & The number of multicast groups \\
~$\mathcal{G}_k$ & The set of users in multicast group $k$\\
~$\mathbf{W}_k$ & The binary feature selection matrix for group $k$\\
~$\hat{R}_k $ & The transmission rate for the multicast group $\mathcal{G}_k$ \\
~$\delta $ &  The data volume of a single BEV grid \\ 
~$T $ & The latency budget \\
~$\alpha_{n,m}$ & The decodability indicator for user $n$ at MCS rate $r_m$ \\
~$\mathbf{X}$ & The binary grid-rate decision matrix \\
~$x_{l,m}$ & The variable indicating the $l$-th grid is encoded at rate $r_m$
\\
 \hline
\end{tabular}
\end{table}

\subsection{Feature Aggregation}
Based on the MoI and the specific channel conditions of user $n$, the HVN transmits a selected set of BEV grids according to the interest-aware multicast scheme described in Sections \ref{sec:problem}-\ref{sec:algorithm}.

For user $n$ assigned to multicast group $\mathcal{G}_k$ (i.e., $n \in \mathcal{G}_k$), the received feature map is masked by the optimized feature-selection matrix $\mathbf{W}_k$. The received data is thus given by:
\begin{equation}
    \mathbf{D}_{\mathrm{HVN}\rightarrow n} = \mathbf{W}_k \odot \mathbf{F}_{\mathrm{HVN}}.
\end{equation}
Upon reception, user $n$ aligns the HVN feature $\mathbf{D}_{\mathrm{HVN}\rightarrow n}$ to its own BEV frame via a geometric transformation. At each spatial location, the ego feature of user $n$ and the aligned HVN feature are stacked into a matrix $\mathbf{Z} \in \mathbb{R}^{2 \times V}$, which is then aggregated using single-head scaled dot-product self-attention mechanism\cite{hu2022where2comm, tao2025directed}:
\begin{equation}
    \mathbf{A}
    = \mathrm{softmax}\!\left(
        \frac{\mathbf{Z} \mathbf{Z}^{\top}}{\sqrt{V}}
      \right),
    \qquad
    \mathbf{Y} = \mathbf{A} \mathbf{Z},
\end{equation}
where $\mathbf{Z}$ is used as the query, key, and value, and the softmax operation is applied row-wise. 
The row in $\mathbf{Y}$ corresponding to the ego feature forms the fused representation for that specific location. By applying this pointwise operation across all locations, the user constructs a complete fused BEV feature map.
Finally, since this feature extraction and fusion process typically operates across a multi-scale backbone, these multi-scale fused features are upsampled to a uniform resolution and concatenated along the channel dimension. This yields the comprehensive, high-fidelity BEV representation for the detection heads.

For readers’ convenience, the important notations in this paper are summarized in Table \ref{table:notations}.

\section{Interest-aware Multicasting Formulation} 
\label{sec:problem}
Building upon the maps of interest developed in Section \ref{sec:framework}, this section formulates the joint optimization problem for multicast grouping and feature selection. We further show the NP-hardness of the problem, motivating the design of our efficient approximation algorithm.

\subsection{Problem Formulation} 
To accommodate the heterogeneous yet overlapping data requirements of users, we propose an interest-aware multicasting scheme. Specifically, the HVN assigns $N$ users into $K$ multicast groups, denoted by the index set $\mathcal{K} = \{1, \dots, K\}$. For each group $\mathcal{G}_k$, the HVN selectively transmits a subset of BEV grids, represented by a binary feature selection matrix $\mathbf{W}_k \in \{0,1\}^{H \times W}$. To maximize the total utility of the system while strictly adhering to the latency budget, we jointly optimize the user grouping $\mathcal{G}_k$, the MCS rate $\hat{R}_k$, and the feature selection matrix $\mathbf{W}_k \in \{0,1\}^{H \times W}$. The corresponding optimization problem is formulated as follows:
\begin{subequations}\label{p:orig_multicast}
\begin{align}
\mathcal{P}_1: \max_{K, \{\mathbf{W}_k, {\mathcal{G}_{k}}, \hat{R}_k \}_{k=1}^K} &\sum_{n=1}^N \left\langle \mathbf{M}_n, \max_{1\le k\le K} \left\{\mathbb{I}(n \in \mathcal{G}_k)\cdot\mathbf{W}_k\right\} \right\rangle  \label{obj:fair_utility} \\
\text { s.t. } ~~~~ & \sum_{k=1}^K \frac{\delta \|\mathbf{W}_k\|_1}{ \hat{R}_k } \leq T, \label{cons:orig_latency}\\
& \hat{R}_k = \min_{n\in \mathcal{G}_k} \{R_n\}, \quad \forall k \in \mathcal{K}, \label{cons:orig_rate_def}\\
& \mathbf{W}_k \in \{0, 1\}^{H \times W}, \quad \forall k \in \mathcal{K}, \label{cons:orig_binary}
\end{align}
\end{subequations}
where $\mathbb{I}(\cdot)$ is an indicator function that equals $1$ if user $n$ is assigned to group $k$, $\langle \cdot, \cdot\rangle$ denotes the Frobenius inner product (element-wise sum of products), $\delta$ represents the data volume of a single BEV grid, and $T$ denotes the latency budget.

The objective function \eqref{obj:fair_utility} aims to maximize the network-wide utility~\cite{lan2010axiomatic}, where the second term $\max_{1\le k\le K} \{\mathbf{W}_k\} $ ensures that if a BEV grid is received multiple times (across different groups), it contributes to the utility gain only once.
Constraint \eqref{cons:orig_latency} guarantees that the total latency determined by allocated feature volume $\|\mathbf{W}_k\|_1$ and effective group rate $\hat{R}_k$ cannot exceed the allowable latency budget $T$.  Constraint \eqref{cons:orig_rate_def} indicates that the transmission rate $\hat{R}_k$ for multicast group $\mathcal{G}_k$ is bottlenecked by the group member with the poorest channel quality. Without loss of generality, we assume that HVN assigns orthogonal time slots to the multicast groups to eliminate interference and simplify the analysis~\cite{xu2022cellular}.

\begin{remark}
    The formulation in $\mathcal{P}_1$ encompasses unicast and broadcast as special cases. 
    Specifically, our formulation reduces to \textit{broadcast} when $K = 1$ and $\mathcal{G}_1 = \mathcal{N}$, and to \textit{unicast} when $|\mathcal{G}_k| = 1$, $\forall k \in \mathcal{K}$. Moreover, it reduces to a traditional multicasting problem if $\mathbf{M}_n$ is identical for all users, in which case optimizing $\{\mathcal{G}_k\}$ and $\mathbf{W}_k$ becomes rather trivial. Consequently, while motivated by CP, $\mathcal{P}_1$ establishes a generalized optimization framework for wireless multicasting.
\end{remark}

\begin{remark}
    In general, $T$ is typically chosen to be a small value (e.g., $20$ ms) for CP in vehicular networks; otherwise, significant changes in vehicular locations would render the feature maps outdated. 
    Therefore, vehicular mobility is negligible and omitted in our formulation in this paper.
\end{remark}

\subsection{NP-Hardness of the Problem}

We establish the following theorem regarding the complexity of the optimization problem.
\begin{theorem} \label{thm:np_hard}
    Problem $\mathcal{P}_1$ is NP-hard.
\end{theorem}

\begin{IEEEproof}
    We prove the NP-hardness of $\mathcal{P}_1$ by showing that the 0-1 knapsack problem, which is well-known to be NP-hard~\cite{bretthauer2002nonlinear}, reduces to an instance of $\mathcal{P}_1$.

    In a standard 0-1 knapsack problem, we assume that there is a set of items $\mathcal{J}$, where each item $j$ has a profit value $v_j$ and a cost $w_j$. The goal is to select a subset of items to maximize total profit $\sum v_j x_j$ subject to a capacity constraint $\sum w_j x_j \leq C$, where $x_j \in \{0,1\}$.
    We construct an instance of $\mathcal{P}_1$ as follows. We limit the search space of the grouping strategy $\mathcal{G}_k$ to a fixed finite set of candidate groups $\mathcal{G} = \{\mathcal{S}_1, \mathcal{S}_2, \dots, \mathcal{S}_Z\}$. Furthermore, we fix the feature selection $\mathbf{W}_k$ for every group to a constant pattern $\mathbf{W}'$. Under these restrictions, the optimization problem reduces to selecting a subset of groups to maximize utility while satisfying the latency budget $T$.
    The mathematical formulation is
    \begin{subequations} \label{p:knap_restricted} 
    \begin{align} 
    \max_{x_j}~& \sum_{n=1}^N  \min\{1, \sum_{j=1}^{Z} \mathbb{I}(n\in \mathcal{S}_j) x_j\} \cdot \left\langle \mathbf{M}_n, \mathbf{W}' \right\rangle \label{eq:knap_obj}\\
    \text{s.t.}~&\sum_{j=1}^{Z} w_j x_j \leq C,
    \label{eq:knap_cap}\\
    &x_{j}\in\{0,1\},~\forall j \in {1, \dots, Z}, \label{eq:knap_bin} 
    \end{align} 
    \end{subequations}
    where $x_j$ is the binary decision variable for selecting group $\mathcal{S}_j$, $\mathbb{I}(n \in \mathcal{S}_j)$ denotes the indicator function, which equals 1 if user $n$ belongs to group $\mathcal{S}_j$ and 0 otherwise, and $w_j = {\delta \|\mathbf{W}'\|_1}/{\min_{n \in \mathcal{S}_j} \{R_n\}}$ represents the latency cost of group $\mathcal{S}_j$.

    We conclude that a restricted instance of $\mathcal{P}_1$ can be shown to be equivalent to a non-linear 0-1 knapsack problem, which is known to be NP-hard \cite{bretthauer2002nonlinear}. It inherently follows that our joint optimization problem $\mathcal{P}_1$ is also NP-hard.
\end{IEEEproof}

The optimization problem $\mathcal{P}_1$ shares the fundamental intractability of the knapsack problem. Its complexity is further compounded by the bilinear coupling between the group rate $\hat{R}_k$ and the feature selection $\mathbf{W}_k$ in the latency constraint. Moreover, the objective function is non-separable due to the redundancy check ($\max_{1 \le k \le K}\{\mathbf{W}_k\}$), which creates structural dependencies across different multicast groups. Due to this interdependency, it is infeasible to decompose the problem into independent subproblems, rendering standard dynamic programming \cite{bellman1966dynamic} computationally inapplicable. The inherent combinatorial hardness motivates us to transform the problem before developing the solution approach.

\section{Problem Transformation}
\label{sec:reformulation}

In this section, we transform the optimization problem $\mathcal{P}_1$ presented in Section \ref{sec:problem} from a grid-rate perspective and show that the optimal solution to the transformed problem corresponds directly to an optimal solution of the original.

Given the NP-hardness of the original problem $\mathcal{P}_1$ established in Theorem \ref{thm:np_hard}, obtaining a globally optimal solution is computationally intractable. Consequently, we seek to develop a low-complexity algorithm with an approximation guarantee by shifting the optimization perspective. 
Instead of explicitly partitioning users into groups, we transform the problem from a \textit{grid-rate perspective}. Let the BEV feature grids be indexed by $\mathcal{L} = \{1, \dots, L\}$, where $L = H \times W$ represents the total number of grids. For each grid $l \in \mathcal{L}$, the decision variable becomes selecting an optimal MCS rate $r_m \in \mathcal{R}$. We capture this by defining a binary decision matrix $\mathbf{X} \in \{0, 1\}^{L \times M}$, where $x_{l,m} = 1$ indicates that the $l$-th BEV grid is encoded at MCS rate $r_m$. 
This formulation imposes an \textit{implicit} grouping strategy: it is assumed that a BEV grid transmitted in rate $r_m$ is multicast to all users whose received SNR meets the corresponding decodability threshold. Specifically, we define the decodability indicator for user $n$ at rate $r_m$ as $\alpha_{n,m} = \mathbb{I}(\gamma_n \geq \gamma_{\mathrm{th}, m})$, and let $\mathbf{a}_n \in \{0,1\}^M$ be the vector of these indicators. Furthermore, let $\mathbf{m}_n \in \mathbb{R}^L$ denote the vectorized MoI mask $\mathbf{M}_n$. Therefore, the transformed optimization problem is given by:
\begin{subequations}\label{p:reform_multicast}
\begin{align}
   \mathcal{P}_2: \max_{\mathbf{X}} ~&
   \sum_{n=1}^N 
      \mathbf{m}_n^\top 
      \min\{\mathbf{1},\, \mathbf{X}\mathbf{a}_n\}
   \label{obj:feature_psf} \\
   \text{s.t.} ~&
   \sum_{l=1}^L \sum_{m=1}^M  
   \frac{\delta x_{l,m}}{B r_m} \leq T, 
   \label{cons:feature_latency}\\
   & \mathbf{X} \in \{0, 1\}^{L \times M}.
   \label{cons:feature_binary}
\end{align}
\end{subequations}
where $\top$ denotes the transpose operation and $\mathbf{1} \in \mathbb{R}^L$ represents an all-ones vector of length $L$. Constraint \eqref{cons:feature_latency} guarantees that the cumulative transmission time across all selected rates strictly adheres to the latency budget $T$. By leveraging these rate-specific decision variables, the complex original joint multicast grouping and feature selection problem is effectively transformed into a grid-rate allocation problem.

To justify the transformation, we develop the following lemma and theorem to show that solving $\mathcal{P}_2$ yields the optimal solution to $\mathcal{P}_1$.

\begin{lemma} \label{lemma:Maximal_Grouping}
    For any optimal solution to $\mathcal{P}_1$, expanding the optimal multicast group $\mathcal{G}_k^\star$ to include all users capable of decoding at its bottleneck MCS rate $r_m$, i.e., setting $\mathcal{G}_k^\star = \mathcal{N}_m$ where $\mathcal{N}_m = \{n \in \mathcal{N} \mid \gamma_n \geq \gamma_{\mathrm{th}, m}\}$, does not compromise the optimality.
\end{lemma}

\begin{IEEEproof}
    It is straightforward to show this because expanding the optimal multicast group to include all users capable of decoding its bottleneck MCS rate $r_m$, i.e., setting $\mathcal{G}_k^{\star} = \mathcal{N}_m$, does not increase the multicast latency but leads to a non-decreasing total utility.
\end{IEEEproof}

\begin{theorem} \label{theorem:equivalence}
\rev{ Problems $\mathcal{P}_1$ and $\mathcal{P}_2$ are equivalent.}
\end{theorem}

\begin{IEEEproof}
First, we demonstrate that any feasible solution to $\mathcal{P}_1$, denoted by $(K, \{\mathbf{W}_k, \mathcal{G}_k, \hat{R}_k\}_{k=1}^K)$, can be mapped to a feasible solution $\mathbf{X}$ for $\mathcal{P}_2$ with the same objective value. For each multicast group $k$, we identify the specific MCS rate index $m_k$ such that group $k$'s minimum rate is $\hat{R}_k = B r_{m_k}$. For every grid coordinate $(h,w)$ selected by group $k$ (i.e., $[\mathbf{W}_k]_{h,w} = 1$), we compute its flattened index $l$ and set $x_{l, m_k} = 1$. We first check the constraint holds: The total transmission time in $\mathcal{P}_2$ is bounded by $\sum_{l=1}^L \sum_{m=1}^M \frac{\delta x_{l,m}}{B r_m} \le \sum_{k=1}^K \frac{\delta \|\mathbf{W}_k\|_1}{\hat{R}_k} \le T$. This is because multiple groups in $\mathcal{P}_1$ might redundantly transmit the same grid at the same rate, whereas $\mathcal{P}_2$ counts unique grid-rate pairs. For the objective function, if a user $n$ receives grid $l$ via group $k$ in $\mathcal{P}_1$, we have $n \in \mathcal{G}_k$, implying that the maximum rate satisfies $R_n \ge \hat{R}_k$ and consequently $\alpha_{n, m_k} = 1$. In $\mathcal{P}_2$, since $x_{l, m_k} = 1$ and $\alpha_{n, m_k} = 1$, the inner product evaluates to $\min\{\mathbf{1}, \mathbf{X} \mathbf{a}_n\} = 1$, indicating reception of grid $l$. Hence, the objective value of $\mathcal{P}_2$ matches that of $\mathcal{P}_1$.

Next, we show the converse: any feasible solution $\mathbf{X}$ to $\mathcal{P}_2$ can be transformed into a feasible solution for $\mathcal{P}_1$ with the same objective value. For solution mapping, we construct $K = M$ multicast groups. For each rate index $m \in \{1, \dots, M\}$, we define a group $\mathcal{G}_m = \{n \in \mathcal{N} \mid \gamma_n \ge \gamma_{\mathrm{th},m}\}$, ensuring every included user can decode at rate $\hat{R}_m = B r_m$. We then map the $m$-th column of $\mathbf{X}$ to the spatial feature selection matrix $\mathbf{W}_m \in \{0,1\}^{H \times W}$, setting $[\mathbf{W}_m]_{h,w} = x_{l,m}$ where $l = (h-1)W + w$. Since each constructed group $m$ operates at rate $\hat{R}_m$, the total transmission time equals the latency cost of $\mathbf{X}$ in $\mathcal{P}_2$: $\sum_{m=1}^M \frac{\delta \|\mathbf{W}_m\|_1}{\hat{R}_m} = \sum_{m=1}^M \sum_{l=1}^L \frac{\delta x_{l,m}}{B r_m} \le T$. Therefore, the latency constraint of $\mathcal{P}_1$ is satisfied. Furthermore, if user $n$ successfully receives grid $l$ in $\mathcal{P}_2$ (meaning there exists an $m$ such that $x_{l,m} = 1$ and $\alpha_{n,m} = 1$), our mapping guarantees that $[\mathbf{W}_m]_{h,w} = 1$ and $n \in \mathcal{G}_m$, ensuring successful reception in $\mathcal{P}_1$. Hence, the objective value of $\mathcal{P}_1$ matches that of $\mathcal{P}_2$.

Since any feasible solution in either formulation can be trivially mapped to a feasible solution in the other with the same objective value, $\mathcal{P}_1$ and $\mathcal{P}_2$ are equivalent.
\end{IEEEproof}

\begin{figure}
    \centering
    \includegraphics[width=1.0\linewidth]{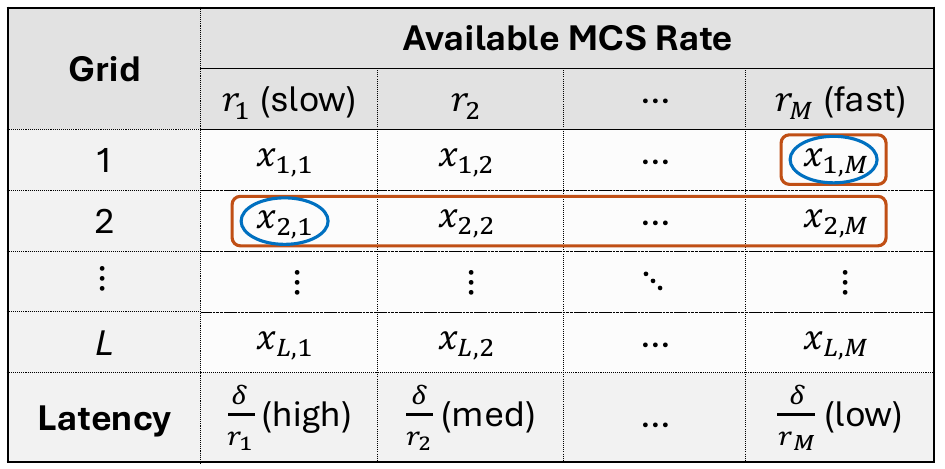}
    \caption{Illustration of the \textit{grid-rate} transformation. The rows represent BEV grids, while the columns correspond to the available MCS rates. The blue circles indicate the selected rates for a given grid. The red rectangle indicates that selecting a specific grid-rate pair also implicitly selects the pairs with higher rates in the same column.}
    \label{table:reformulation}
    \vspace{-0.0cm}
\end{figure}

By the equivalence shown in Theorem \ref{theorem:equivalence}, the transformed problem $\mathcal{P}_2$ is also NP-hard, since an exact solution to $\mathcal{P}_2$ immediately yields an exact solution to $\mathcal{P}_1$. Despite this computational intractability, the structure of $\mathcal{P}_2$ enables us to develop an efficient approximation algorithm. To facilitate approximate solution finding, we visualize the \textit{grid-rate} transformation in Fig. \ref{table:reformulation}. The rows correspond to BEV grids ($1, \dots, L$), while the columns enumerate the available MCS rates, ranging from $r_1$ (slowest) to $r_M$ (fastest). The element $x_{l,m}$ represents the binary decision variable for assigning rate $r_m$ to grid $l$. Notably, a grid is skipped if no MCS rate is selected for it. This formulation highlights an inherent system trade-off: selecting a robust, low rate like $r_1$ maximizes user accessibility while incurring more latency penalty. Conversely, choosing a higher rate like $r_M$ minimizes latency but yields lower system utility.  


\section{Algorithm Design}
\label{sec:algorithm}
In this section, we present the algorithm design for the transformed problem $\mathcal{P}_2$ presented in Section \ref{sec:reformulation}. Although $\mathcal{P}_2$ is still computationally intractable, we can leverage it to find an approximate solution. By establishing the submodularity of the problem, we develop an efficient algorithm with a provable approximation guarantee.

\subsection{Submodularity Analysis}
\label{subsec:submodularity}

Our transformed problem $\mathcal{P}_2$ transforms the original problem into a set optimization problem. In particular, it aims to select an optimal subset from a ground set to maximize an objective function that, as we show below, is monotone submodular.

\begin{definition}[Monotone Submodularity]\label{def:submodular}
Let $\mathcal{E}$ be a finite ground set and let $f: 2^{\mathcal{E}} \rightarrow \mathbb{R}_{\geq 0}$ be a set function. 
For any subset $\mathcal{S} \subseteq \mathcal{E}$ and element $e \in \mathcal{E} \setminus \mathcal{S}$, define the marginal gain of $e$ with respect to $\mathcal{S}$ as
\begin{equation}
    \Delta_e f(\mathcal{S}) \triangleq f(\mathcal{S} \cup \{e\}) - f(\mathcal{S}).
\end{equation}
The function $f$ is said to be \emph{submodular} if it satisfies the diminishing-returns property
\begin{equation}
    \Delta_e f(\mathcal{A}) \geq \Delta_e f(\mathcal{B}), \quad \forall\, \mathcal{A} \subseteq \mathcal{B} \subseteq \mathcal{E}, \; e \in \mathcal{E} \setminus \mathcal{B}.
\end{equation}
If, in addition, $f$ is monotone non-decreasing, i.e., $f(\mathcal{A}) \leq f(\mathcal{B})$ for all $\mathcal{A} \subseteq \mathcal{B} \subseteq \mathcal{E}$, then $f$ is called a \emph{monotone submodular} function.
\end{definition}

\begin{theorem} \label{theo:submodular}
    Problem $\mathcal{P}_2$ is a monotone submodular maximization problem subject to a knapsack constraint.
\end{theorem}

\begin{IEEEproof}
We show that the objective of problem~$\mathcal{P}_2$ is a monotone submodular function.
Let the ground set of all possible grid-rate pairs be defined as
\begin{equation}
    \mathcal{E} \triangleq \{(l,m) \mid l \in \mathcal{L}\}, m \in \{1,\ldots,M\}\}.
\end{equation}
For any subset $\mathcal{S} \subseteq \mathcal{E}$, we construct the binary decision matrix $\mathbf{X}(\mathcal{S}) \in \{0,1\}^{L \times M}$ such that $x_{l,m} = 1$ if and only if $(l,m) \in \mathcal{S}$.
For user $n$, let 
\begin{equation}
    \mathbf{y}_n(\mathcal{S}) \triangleq \min\{\mathbf{1}, \mathbf{X}(\mathcal{S})\mathbf{a}_n\} \in \{0,1\}^L
\end{equation}
where the $l$-th element, $y_{n,l}(\mathcal{S})$, equals $1$ if user $n$ receives BEV feature grid $l$ at a decodable rate selected in $\mathcal{S}$, and $0$ otherwise. 
In other words, $y_{n,l}(\mathcal{S}) = 1$ if there exists at least one pair $(l,m) \in \mathcal{S}$ such that $a_{n,r} = 1$.

The objective function \eqref{obj:feature_psf} can thus be expressed as a set function $f: 2^{\mathcal{E}} \rightarrow \mathbb{R}_{\geq 0}$
\begin{equation}
    f(\mathcal{S}) \triangleq \sum_{n=1}^N \mathbf{m}_n^\top \mathbf{y}_n(\mathcal{S}) = \sum_{n=1}^N \sum_{l=1}^L m_{n,l} y_{n,l}(\mathcal{S})
\end{equation}
where $m_{n,l} \geq 0$ is the $l$-th element of the MoI vector $\mathbf{m}_n$.

We now prove that $f(\mathcal{S})$ is monotone submodular.

\textit{Monotonicity.}
Consider two sets $\mathcal{A} \subseteq \mathcal{B} \subseteq \mathcal{E}$. If a grid $l$ is decodable for user $n$ under subset $\mathcal{A}$, it remains decodable under $\mathcal{B}$. Therefore, $y_{n,l}(\mathcal{A}) \leq y_{n,l}(\mathcal{B})$ for all $n$ and $l$. Since the weights are non-negative ($m_{n,l} \geq 0$), it directly follows that $f(\mathcal{A}) \leq f(\mathcal{B})$. Thus, $f(\mathcal{S})$ is monotonically non-decreasing.


\textit{Submodularity.}
To prove submodularity, we must show that the marginal gain diminishes as the set grows. Let $\mathcal{A} \subseteq \mathcal{B} \subseteq \mathcal{E}$, and consider a new element $e = (l', m') \in \mathcal{E} \setminus \mathcal{B}$. 
The marginal gain of adding $e$ to a set $\mathcal{S}$ is defined as 
\begin{equation}
    \Delta_e f(\mathcal{S}) = f(\mathcal{S} \cup \{e\}) - f(\mathcal{S}).
\end{equation}

Adding $e = (l', m')$ only potentially changes the coverage status of grid $l'$, and only for users who can decode at rate $m'$, i.e., users with $a_{n,m'} = 1$. For any user $n$ capable of decoding at rate $m'$, the increase in the objective function is $m_{n,l'}$ if grid $l'$ was not previously covered, and $0$ if it was already covered. Thus, the total marginal gain for a set $\mathcal{S}$ is
\begin{equation}
    \Delta_e f(\mathcal{S}) = \sum_{n: a_{n,m'} = 1} m_{n,l'} \big(1 - y_{n,l'}(\mathcal{S})\big).
\end{equation}
Since $\mathcal{A} \subseteq \mathcal{B}$, we already have $y_{n,l'}(\mathcal{A}) \leq y_{n,l'}(\mathcal{B})$. Consequently, it holds that
\begin{equation}
    1 - y_{n,l'}(\mathcal{A}) \geq 1 - y_{n,l'}(\mathcal{B}).
\end{equation}
Multiplying both sides by the non-negative weight $m_{n,l'}$ and summing over all applicable users $n$ yields
\begin{equation}
    \Delta_e f(\mathcal{A}) \geq \Delta_e f(\mathcal{B}).
\end{equation}
Therefore, $f(\mathcal{S})$ is a monotone submodular function. Considering (\ref{cons:feature_latency}) is a knapsack constraint, we can conclude that Problem $\mathcal{P}_2$ is a monotone submodular maximization subject to a knapsack constraint. 
This completes the proof.

\end{IEEEproof}

\subsection{Greedy Algorithm\label{subsec:standard_greedy}}
It is well known that the greedy procedure is highly effective in solving monotone submodular maximization problems \cite{khuller1999budgeted}. To this end, we design a greedy procedure, which prioritizes transmissions providing the highest marginal utility over unit latency cost. We first define the ground set of all candidate transmission items as
\begin{equation}
\mathcal{E} \triangleq \{(l,m) \mid l \in \mathcal{L}, m\in \{1,2,\ldots, M\}\},
\end{equation}
where each item $e =(l,m)$ corresponds to transmitting the $l$-th BEV feature grid at rate $r_m$. Let $\mathcal{S} \subseteq \mathcal{E}$ denote the set of items that have already been selected in the current iteration. To evaluate the total utility, we define the equivalent function $F: 2^{\mathcal{E}} \rightarrow \mathbb{R}_{\geq 0}$, corresponding to the objective function in $\mathcal{P}_2$, as:
\begin{equation}
F(\mathcal{S}) = \sum_{n=1}^N \sum_{l=1}^L [\mathbf{m}_n]_l \cdot \min \left\{ 1, \sum_{(l, m') \in \mathcal{S}} \alpha_{n,m'} \right\}.
\end{equation}
The marginal gain of adding a candidate item $e \in \mathcal{E} \setminus \mathcal{S}$ is given by:
\begin{equation}
\Delta_e F(\mathcal{S}) = F(\mathcal{S}\cup \{e\})-F(\mathcal{S}).  
\end{equation}
Moreover, we define the marginal utility ratio $\rho_e(\mathcal{S})$ as:
\begin{equation}
\rho_e(\mathcal{S}) \triangleq \frac{\Delta_e F(\mathcal{S})}{c_e},
\end{equation}
where $c_e = \frac{\delta}{B r_m}$ represents the transmission latency of item $e$. In each iteration, the algorithm greedily selects the item $e^\star$ that maximizes this ratio. This process continues iteratively until the latency budget $T$ is exhausted, yielding the resulting set $\mathcal{S}_{\mathrm{greedy}}$.

If a grid $l$ is transmitted at two distinct rates $r_m < r_{m'}$, the higher-rate transmission becomes redundant. Consequently, an optimal decision matrix $\mathbf{X}^\star$ satisfies $\sum_m x_{l,m} \leq 1$ for every $l \in \mathcal{L}$. The standard greedy procedure, however, may first select a high-rate item $(l,m')$ for its favorable marginal utility ratio and later add a lower-rate item $(l,m)$ to serve additional users, thereby incurring a redundant cost $c_{(l,m')}$. 
To eliminate this inefficiency, we remove some grid-rate selections. Specifically, for every grid $l$ assigned multiple rates, we retain only the lowest-rate item to recover redundant latency costs. Then, since the latency decreases, we involve the same greedy procedure to select new grid-rate items. We perform this consolidation and reinvestment step once to fill the remaining latency budget $T$. 
This process leads to non-decreasing total utility, ultimately yielding the refined solution $\mathcal{S}_{\mathrm{greedy}}^\dagger$. Therefore, we have $
    F(\mathcal{S}_{\mathrm{greedy}}^\dagger) \geq  F(\mathcal{S}_{\mathrm{greedy}}).$

To address the limitation of greedy algorithms for high-cost items, we incorporate a ``best-single-item'' check \cite{sviridenko2004note}. In particular, the final solution $\mathcal{S}^\star$ is obtained by comparing the greedy solution $\mathcal{S}_{\mathrm{greedy}}^\dag$ against the best feasible single item $\mathcal{S}_{\mathrm{single}}$, i.e.,
\begin{equation}
\mathcal{S}^\star = {\arg \max}_{\mathcal{S} \in \{\mathcal{S}_{\mathrm{greedy}}^\dag, \mathcal{S}_{\mathrm{single}}\}} F(\mathcal{S}),
\end{equation}
where $\mathcal{S}_{\mathrm{single}} = {\arg \max}_{\{e \in \mathcal{E} \mid c_e \leq T\}} F(\{e\})$. 
Finally, we reconstruct the optimal binary decision matrix $\mathbf{X}^\star$ from the final solution $\mathcal{S}^\star$ as follows
\begin{equation} \label{eq:reconstruct_x}
    x_{l,m}^\star=
    \begin{cases}
        1,&\text{if }(l,m)\in\mathcal{S}^\star\\0,&\text{otherwise}
    \end{cases}
\end{equation}
Then, we map the obtained $\mathbf{X}^\star$ to the solution to the original problem $\mathcal{P}_1$ as
\begin{equation} 
\begin{aligned}
&K^\star = M,\\
&\mathcal{G}_k^\star = \{n \in \mathcal{N} \mid \alpha_{n,k} = 1\},  \quad \forall k  \in \{1,\dots,K^\star\},\\
&[\mathbf{W}_k^\star]_{h,w} = \mathbf{X}_{(h-1)W+w,\, k}^\star, \\ 
& \quad \quad \quad \quad \quad \quad \forall h \in \{1,\dots,H\}, \forall w \in \{1,\dots,W\},\\
& \hat{R}_k^\star = \min_{n \in \mathcal{G}_k^\star} \{R_n\}, \quad \forall k \in \{1,\dots, K^\star\}.
\end{aligned} \label{eq:map_back} 
\end{equation}
The complete procedure is detailed in Algorithm \ref{alg:greedy}.

\begin{algorithm}[t]
    \caption{Refined Greedy Algorithm}
    \label{alg:greedy}
    \LinesNumbered
    \KwIn{$\mathcal{E}$, $T$, $c_e, \forall e \in \mathcal{E}$}
    \KwOut{$K^\star, \mathcal{G}_k^\star, \mathbf{W}_k^\star$, $\hat{R}_k^\star$}
    
    Initialize $\mathcal{S}_{\mathrm{greedy}}^\dag \leftarrow \emptyset$, $T_{\mathrm{rem}} \leftarrow T$, $\mathcal{C} \leftarrow \mathcal{E}$\;
    
    \For{$iter \leftarrow 1$ \KwTo $2$}{
        
        \tcp{Greedy Selection}
        \While{$\mathcal{C} \neq \emptyset$ \textnormal{\textbf{and}} $T_{\mathrm{rem}} > 0$}{
            $e^\star \leftarrow \operatorname*{argmax}_{e \in \mathcal{C}} \frac{F(\mathcal{S}_{\mathrm{greedy}}^\dag \cup \{e\}) - F(\mathcal{S}_{\mathrm{greedy}}^\dag)}{c_e}$\;
            
            \If{$F(\mathcal{S}_{\mathrm{greedy}}^\dag \cup \{e^\star\}) - F(\mathcal{S}_{\mathrm{greedy}}^\dag) \leq 0$}{
                \textbf{break}\;
            }
            
            \If{$c_{e^\star} \leq T_{\mathrm{rem}}$}{
                $\mathcal{S}_{\mathrm{greedy}}^\dag \leftarrow \mathcal{S}_{\mathrm{greedy}}^\dag \cup \{e^\star\}$\;
                $T_{\mathrm{rem}} \leftarrow T_{\mathrm{rem}} - c_{e^\star}$\;
            }
            $\mathcal{C} \leftarrow \mathcal{C} \setminus \{e^\star\}$\;  
        }

        \If{$iter = 1$}{
            \tcp{Redundant item removal}
            $\Omega \leftarrow \bigcup_{l \in \mathcal{L}} \big\{ (l,m) \in \mathcal{S}_{\mathrm{greedy}}^\dag \mid m \neq \operatorname*{argmin}_{\{m' \mid (l,m') \in \mathcal{S}_{\mathrm{greedy}}^\dag\}} r_{m'} \big\}$\;
            $\mathcal{S}_{\mathrm{greedy}}^\dag \leftarrow \mathcal{S}_{\mathrm{greedy}}^\dag \setminus \Omega$\;
            $\Delta T \leftarrow \sum_{e \in \Omega} c_e$\;
            
            \If{$\Delta T \le 0$}{
                \textbf{break}\;
            }
            $T_{\mathrm{rem}} \leftarrow T_{\mathrm{rem}} + \Delta T$\;
            $\mathcal{C} \leftarrow \mathcal{E} \setminus \mathcal{S}_{\mathrm{greedy}}^\dag$\;
        }
    }
    
    \tcp{Best-single-item Check \& Output}
    $\mathcal{S}_{\mathrm{single}} = \operatorname{argmax}_{\{e \in \mathcal{E} \mid c_e \leq T\}} F(\{e\})$\;
    $\mathcal{S}^\star \leftarrow \operatorname*{argmax}_{\mathcal{S} \in \{\mathcal{S}_{\mathrm{greedy}}^\dag, \mathcal{S}_{\mathrm{single}}\}} F(\mathcal{S})$\;
    Reconstruct $\mathbf{X}^\star$ from $\mathcal{S}^\star$ via \eqref{eq:reconstruct_x}\;
    Map the solution $K^\star, \mathcal{G}_k^\star, \mathbf{W}_k^\star, \hat{R}_k^\star$ to $\mathcal{P}_1$ via \eqref{eq:map_back}\;
    
    \Return $K^\star, \mathcal{G}_k^\star, \mathbf{W}_k^\star, \hat{R}_k^\star$
\end{algorithm}

\begin{theorem}
    The proposed algorithm \ref{alg:greedy} yields a $(1-1/\sqrt{e})$-approximate solution to $\mathcal{P}_1$ and $\mathcal{P}_2$ in the worst case.
\end{theorem}

\begin{IEEEproof}
    Since the transformed problem $\mathcal{P}_2$ is an instance of monotone submodular maximization subject to a knapsack constraint, Khuller \textit{et al.} \cite{khuller1999budgeted} have shown that selecting the better of the standard greedy solution and the single best item guarantees a theoretical lower bound of 
    \begin{equation}
        \max \{F(\mathcal{S}_{\mathrm{greedy}}), F(\mathcal{S}_{\mathrm{single}})\} \geq \left(1 - \frac{1}{\sqrt{e}}\right) F(\mathcal{S}^{\mathrm{OPT}}),
    \end{equation}
    where $\mathcal{S}^{\mathrm{OPT}}$ denotes the optimal set in $\mathcal{P}_2$.
    
    Our proposed approach augments the standard greedy procedure with an in-grid rate consolidation mechanism, yielding the refined solution $\mathcal{S}_{\mathrm{greedy}}^\dag$, which satisfies $F(\mathcal{S}_{\mathrm{greedy}}^\dag) \geq F(\mathcal{S}_{\mathrm{greedy}})$. Recall that our final solution is $\mathcal{S}^\star = {\arg \max}_{\mathcal{S} \in \{\mathcal{S}_{\mathrm{greedy}}^\dag, \mathcal{S}_{\mathrm{single}}\}} F(\mathcal{S})$, and it directly follows that
    \begin{equation}
        \begin{aligned}
            F(\mathcal{S}^\star) &= \max \left\{ F(\mathcal{S}_{\mathrm{greedy}}^\dag), F(\mathcal{S}_{\mathrm{single}}) \right\} \\
            &\geq \max \left\{ F(\mathcal{S}_{\mathrm{greedy}}), F(\mathcal{S}_{\mathrm{single}}) \right\} \\
            &\geq \left(1 - \frac{1}{\sqrt{e}}\right) F(\mathcal{S}^{\mathrm{OPT}}).       
        \end{aligned}
    \end{equation}
    Furthermore, since $\mathcal{P}_1$ and $\mathcal{P}_2$ are equivalent, the approximate ratio also applies to $\mathcal{P}_1$, which concludes the proof.
\end{IEEEproof}

\begin{remark}
    It is important to note that the theoretical bound of $1-1/\sqrt{e}$ occurs only in worst-case scenarios. In practical large-scale systems where the budget $T$ is large relative to the cost of individual items ($c_e \ll T$), the performance of the greedy algorithm improves significantly, asymptotically approaching the tighter $1 - 1/e \approx 0.63$ guarantee, and often nearing the global optimum when the objective function exhibits low curvature. 
\end{remark}

\emph{Computational Complexity:}
The ground set contains $|\mathcal{E}| = L \times M$ candidate items, where $L = H \times W$ is the total number of grids. In the worst case, the greedy algorithm iterates up to $|\mathcal{E}|$ times, computing the marginal gain for at most $|\mathcal{E}|$ items in each iteration. Moreover, evaluating the marginal gain $\Delta_e F(\mathcal{S})$ requires a summation over the $N$ users for a specific grid $l$, resulting in a complexity of $O(N)$ per item evaluation. The greedy algorithm, therefore, requires $O(N \cdot |\mathcal{E}|^2)$ operations. 
Furthermore, the subsequent reinvestment phase invokes the greedy procedure again, keeping the overall worst-case complexity of $O(N \cdot |\mathcal{E}|^2)$.
Lastly, the ``best-single-item'' check requires $O(N \cdot |\mathcal{E}|)$ operations. Consequently, the overall computational complexity is $O(N \cdot |\mathcal{E}|^2)$ in the worst case.

\begin{algorithm}[t]
    \caption{Accelerated Greedy Algorithm}
    \label{alg:lazy_greedy}
    \LinesNumbered
    \SetKwFunction{ExtractMax}{ExtractMax}
    \SetKwFunction{MaxPrior}{MaxPrior}
    \KwIn{$\mathcal{E}$, $T$, $c_e, \forall e \in \mathcal{E}$}
    \KwOut{$K^\star, \mathcal{G}_k^\star, \mathbf{W}_k^\star, \hat{R}_k^\star$}
    
    Initialize $\mathcal{S}_{\mathrm{greedy}}^\dag \leftarrow \emptyset$, $T_{\mathrm{rem}} \leftarrow T$, $\mathcal{Q} \leftarrow \emptyset$\;

    \For{$e \in \mathcal{E}$}{
        $\rho_{e, \mathrm{init}} \leftarrow F(\{e\}) / c_e$\;
        \lIf{$c_e \leq T$ \textnormal{\textbf{and}} $\rho_{e, \mathrm{init}} > 0$}{$\mathcal{Q}.\texttt{Insert}(e, \rho_{e,\mathrm{init}})$}
    }
    
    \For{$iter \leftarrow 1$ \KwTo $2$}{\tcp{Accelerated Greedy Selection}
        \While{$\mathcal{Q} \neq \emptyset$ \textnormal{\textbf{and}} $T_{\mathrm{rem}} > 0$}{
            $(e, \rho_{e, \mathrm{bound}}) \leftarrow \mathcal{Q}.\ExtractMax()$\;
            
            \lIf{$c_e > T_{\mathrm{rem}}$}{\textbf{continue}}
            
            $\rho_{e, \mathrm{curr}} \leftarrow \frac{F(\mathcal{S}_{\mathrm{greedy}}^\dag \cup \{e\}) - F(\mathcal{S}_{\mathrm{greedy}}^\dag)} {c_e}$\;
            
            \eIf{$\mathcal{Q} = \emptyset$ \textnormal{\textbf{or}} $\rho_{e, \mathrm{curr}} \geq \mathcal{Q}.\MaxPrior()$}{
                \lIf{$\rho_{e, \mathrm{curr}} \leq 0$}{\textbf{break}}
                $\mathcal{S}_{\mathrm{greedy}}^\dag \leftarrow \mathcal{S}_{\mathrm{greedy}}^\dag \cup \{e\}$ \;
                $T_{\mathrm{rem}} \leftarrow T_{\mathrm{rem}} - c_e$\;
                Let $e = (l,m)$; $\mathcal{Q} \leftarrow \mathcal{Q} \setminus \{ (l,m') \in \mathcal{Q} \mid m' > m\}$\;
            }{
                $\mathcal{Q}.\texttt{Insert}(e, \rho_{e,\mathrm{curr}})$\;
            }
        }
        
        \If{$iter = 1$}{
        \tcp{Redundant item removal}
            $\Omega \leftarrow \bigcup_{l \in \mathcal{L}} \big\{ (l,m) \in \mathcal{S}_{\mathrm{greedy}}^\dag \mid m \neq \operatorname*{argmin}_{\{m' \mid (l,m') \in \mathcal{S}_{\mathrm{greedy}}^\dag\}} r_{m'} \big\}$\;
            $\mathcal{S}_{\mathrm{greedy}}^\dag \leftarrow \mathcal{S}_{\mathrm{greedy}}^\dag \setminus \Omega$\;
            $\Delta T \leftarrow \sum_{e \in \Omega} c_e$\;
            
            \lIf{$\Delta T \leq 0$}{\textbf{break}}
            $T_{\mathrm{rem}} \leftarrow T_{\mathrm{rem}} + \Delta T$, $\mathcal{Q} \leftarrow \emptyset$\;
            \For{$e \in \mathcal{E} \setminus \mathcal{S}_{\mathrm{greedy}}^\dag$}{
                $\rho_{e, \mathrm{curr}} \leftarrow \frac{F(\mathcal{S}_{\mathrm{greedy}}^\dag \cup \{e\}) - F(\mathcal{S}_{\mathrm{greedy}}^\dag)}{c_e}$\;
                \lIf{$c_e \leq T_{\mathrm{rem}}$ \textnormal{\textbf{and}} $\rho_{e, \mathrm{curr}} > 0$}{$\mathcal{Q}.\texttt{Insert}(e, \rho_{e,\mathrm{curr}})$}
            }
        }
    }
    \tcp{Best-single-item Check \& Output}       
    $\mathcal{S}_{\mathrm{single}} = \operatorname{argmax}_{\{e \in \mathcal{E} \mid c_e \leq T\}} F(\{e\})$\;
    $\mathcal{S}^\star \leftarrow \operatorname*{argmax}_{\mathcal{S} \in \{\mathcal{S}_{\mathrm{greedy}}^\dag, \mathcal{S}_{\mathrm{single}}\}} F(\mathcal{S})$\;
    Reconstruct $\mathbf{X}^\star$ from $\mathcal{S}^\star$ via \eqref{eq:reconstruct_x}\;
    Map the solution $K^\star, \mathcal{G}_k^\star, \mathbf{W}_k^\star, \hat{R}_k^\star$ to $\mathcal{P}_1$ via \eqref{eq:map_back}\;
    \Return $K^\star, \mathcal{G}_k^\star, \mathbf{W}_k^\star, \hat{R}_k^\star$
\end{algorithm}

\subsection{Accelerated Greedy Algorithm}
\label{subsec:lazy_greedy}
The standard greedy strategy, i.e., Algorithm \ref{alg:greedy}, necessitates re-evaluating the marginal utility ratio $\rho_e(\mathcal{S})$ for all remaining candidate items $e \in \mathcal{E} \setminus \mathcal{S}$ at every iteration, which may be prohibitive for real-time applications. To mitigate this, we employ an \textit{accelerated} greedy algorithm by exploiting the monotone submodularity property established in Theorem \ref{theo:submodular} and the special structure of our candidate set to significantly reduce the number of evaluations required at each iteration. 

The first modification comes from \cite{minoux2005accelerated}. The monotone submodularity property implies that the marginal gain of adding any item $e$ is non-increasing as the selected set grows, which implies that there can be redundancy in item evaluation. Since the transmission latency $c_e$ is constant for a given item, the marginal utility ratio $\rho_e(\mathcal{S}) \triangleq {\Delta_e F(\mathcal{S})}/{c_e}$ is inherently non-increasing. Therefore, any value of $\rho_{e}(\mathcal{S})$ computed in a prior iteration is an upper bound for the current one. Underpinned by this observation, as shown in Algorithm \ref{alg:lazy_greedy}, we maintain candidate items in a max-priority queue $\mathcal{Q}$, which is sorted by their most recently computed marginal utility ratios. 
In each iteration, we extract the candidate $e$ with the highest upper bound and re-evaluate its current marginal utility ratio $\rho_{e, \mathrm{curr}}$ relative to the current set (lines 8-10). If this updated value exceeds or equals the upper bound of the next candidate in $\mathcal{Q}$, submodularity mathematically guarantees that $e$ provides the maximum marginal ratio among all remaining items. Thus, it is immediately selected without evaluating the remaining items; otherwise, $e$ is re-inserted into $\mathcal{Q}$ with its updated ratio for future consideration (lines 11-18).

The second modification comes from the grid-rate structure in Fig.~\ref{table:reformulation}. As shown in the figure, for a given grid $l$, transmitting it at a lower rate (toward the left, with higher latency) is decodable by a superset of users compared to any higher-rate option (toward the right). Hence, once an item $(l,m)$ is selected in a given round, all remaining candidate items $(l, m')$ with the same grid $l$ and a higher rate $m'>m$ can be discarded. This is because these items will no longer improve the total utility in all subsequent greedy selections.

After obtaining the solution, we further refine it using the redundant item removal process and the ``best-single-item'' check, which is the same as Algorithm \ref{alg:greedy}. Finally, we reconstruct the optimal solution $K^\star, \mathcal{G}_k^\star, \mathbf{W}_k^\star, \hat{R}_k^\star$ to problem $\mathcal{P}_1$, similar to the procedure in Section \ref{subsec:standard_greedy}.
The complete procedure is detailed in Algorithm \ref{alg:lazy_greedy}. We establish the following theorem.

\begin{theorem}
    The proposed Algorithm \ref{alg:lazy_greedy} yields a $(1-1/\sqrt{e})$-approximate solution to $\mathcal{P}_1$ and $\mathcal{P}_2$ in the worst case.
\end{theorem}
\begin{IEEEproof}
The first modification has been shown to find the exact solution as the standard greedy algorithm~\cite{minoux2005accelerated}, and the second modification only removes items with no utility gains without degrading the performance. As a result, it is easy to show it achieves the same worst-case approximation ratio as Algorithm \ref{alg:greedy}.
\end{IEEEproof}

\begin{remark}
    While advanced optimization schemes such as partial enumeration \cite{sviridenko2004note} or Greedy+Max \cite{yaroslavtsev2020bring} can theoretically improve the approximation ratio to $1 - 1/e$ or $1/2$ respectively, they achieve this at the expense of prohibitive computational complexity (e.g., scaling as $\mathcal{O}(|\mathcal{E}|^5)$). Such overhead is intractable for large-scale networks. In contrast, our proposed strategy provides a provable $1 - 1/\sqrt{e}$ guarantee with significantly lower computational overhead, making it well-suited for latency-sensitive applications, e.g., CP in autonomous driving.
\end{remark}

\emph{Computational Complexity:}
It is easy to verify that the worst-case complexity of the accelerated greedy algorithm is the same as that of Algorithm \ref{alg:greedy}. However, by avoiding redundant re-evaluations, the proposed accelerated greedy algorithm achieves enormous speedups in practice.

\section{Experiments}\label{sec:simulation}
In this section, we provide numerical experiments to evaluate the performance of our proposed Birdcast framework. We compare our Birdcast framework with several benchmark methods to demonstrate its superiority.

\subsection{Experiment Settings}

\textbf{Setting.}
While the Birdcast framework is broadly applicable to wireless transmission scenarios featuring heterogeneous user interests and channel conditions, we instantiate and evaluate it in the context of V2I-CP, which serves as the primary motivating application throughout this work. We evaluate the framework on the V2X-Sim dataset \cite{li2022v2x}, a comprehensive V2X CP dataset generated via SUMO and CARLA. We randomly select 10 scenes of data for performance evaluation.
The communication operates at a center frequency of 5.9 GHz with a bandwidth of $B = 100$ MHz \cite{5gaa2021deployment}. The channel model assumes a path-loss exponent of 3.8, and the maximum latency budget is set to $T = 30$ ms \cite{ma2025raise}.
The discrete set of MCS rates (in bits/s/Hz) is:
$\mathcal{R} = \{0.31, 0.49, 0.74, 1.03, 1.33, 1.48, 1.91, 2.41, 2.57, 3.03, 3.61, \\
4.21, 4.82, 5.33\}$  \cite{3GPPTR,feukeu2013mcs}.
The corresponding SNR thresholds required for reliable decoding are: $\Gamma = \{-4.0, -1.0, 2.5, 5.5, 8.5, 10.5, 13.0, 16.0, 18.0, 20.5, 24.0, 27.0, \\ 30.0, 33.0\}$ dB. 
The HVN transmit power is 23 dBm, and the noise power is -92 dBm. For the 3D object detection task, we employ PointPillars \cite{lang2019pointpillars} as the backbone detector. 
The size of BEV grids is $L =  250$, and the data volume of a single grid is $\delta = 1.6$ KB. The window size for the information mask calculation is $W_\mathrm{s} = 5$.
The setup for our experiments includes 2 Intel(R) Xeon(R) Silver 4410Y CPUs (2.0 GHz), 4 NVIDIA RTX A5000 GPUs, and 512 GB DDR4 RAM.
The key parameters are summarized in Table \ref{table_para}.
We evaluate both total utility and detection performance, measured by mean average precision (mAP), when sharing selected features under communication constraints.

\begin{table}[t]
\centering
\caption{Parameter settings for simulations.}
\label{table_para}
\renewcommand{\arraystretch}{1.4}
\setlength{\tabcolsep}{2mm}
\begin{tabular}{|c|c|}
\hline
{Number of users} &  {$N$ = 20-25}\\  \hline
{Bandwidth}& {$B = 100$ MHz}   \\  \hline
Latency Budget  & {$T = 30$ ms}  \\  \hline
The discrete set of MCS   & $\mathcal{R} = \{0.31,  \ldots, 5.33\} $ bits/s/Hz \\ \hline
SNR thresholds & $\Gamma = \{-4.0,  \ldots, 33.0\}$ dB \\  \hline
The number of BEV grids& {$L= 250$ } \\ \hline    
The data volume of a grid & {$\delta = 1.6$ KB } \\ \hline    
The window size  & {\(W_\mathrm{s} = 5\)}  \\  \hline
\end{tabular}
\end{table}

\textbf{Baselines.}
Since existing multicast algorithms primarily focus on user partitioning and resource allocation without natively supporting content (feature) selection, we pair them with a greedy selection strategy. Specifically, once groups and their bottleneck rates are determined, the HVN greedily transmits the feature grids that offer the highest sum of utility until the latency budget $T$ is exhausted. 

\begin{itemize}

    \item \textbf{Marginal-utility \cite{lin2012video}:} A greedy heuristic that dynamically selects features and their transmission rates to maximize the incremental visual quality per unit of time.
    
    \item \textbf{K-means++ \cite{zhang2021joint}:} A clustering-based algorithm that partitions users into groups strictly by channel conditions using modified K-means++. Then, it greedily selects features to maximize local utility given the fixed multicast groups.

    \item \textbf{Dynamic Programming (DP) \cite{chen2015fair}:} A grouping algorithm that sorts users based on their channel qualities and employs a recurrence relation to determine the optimal user partitioning to maximize total utility. We evaluate two variants: i) \textbf{DP-fair}, which imposes a constant fairness constraint as originally in \cite{chen2015fair}, and ii) \textbf{Standard DP}, which omits this constraint.

    \item \textbf{Broadcast \cite{shi2024soar}:} This method aggregates all users into a single group. To ensure reliability, the transmission rate is constrained by the user experiencing the worst channel condition.
    
    \item \textbf{Unicast:} This approach serves each user individually via dedicated links. While it allows for optimal feature transmission per-user link, the latency budget severely limits the volume of features each individual user can receive.

\end{itemize}

\begin{figure}[t]
\centering
\begin{subfigure}[b]{0.24\textwidth} 
  \includegraphics[width=\textwidth]{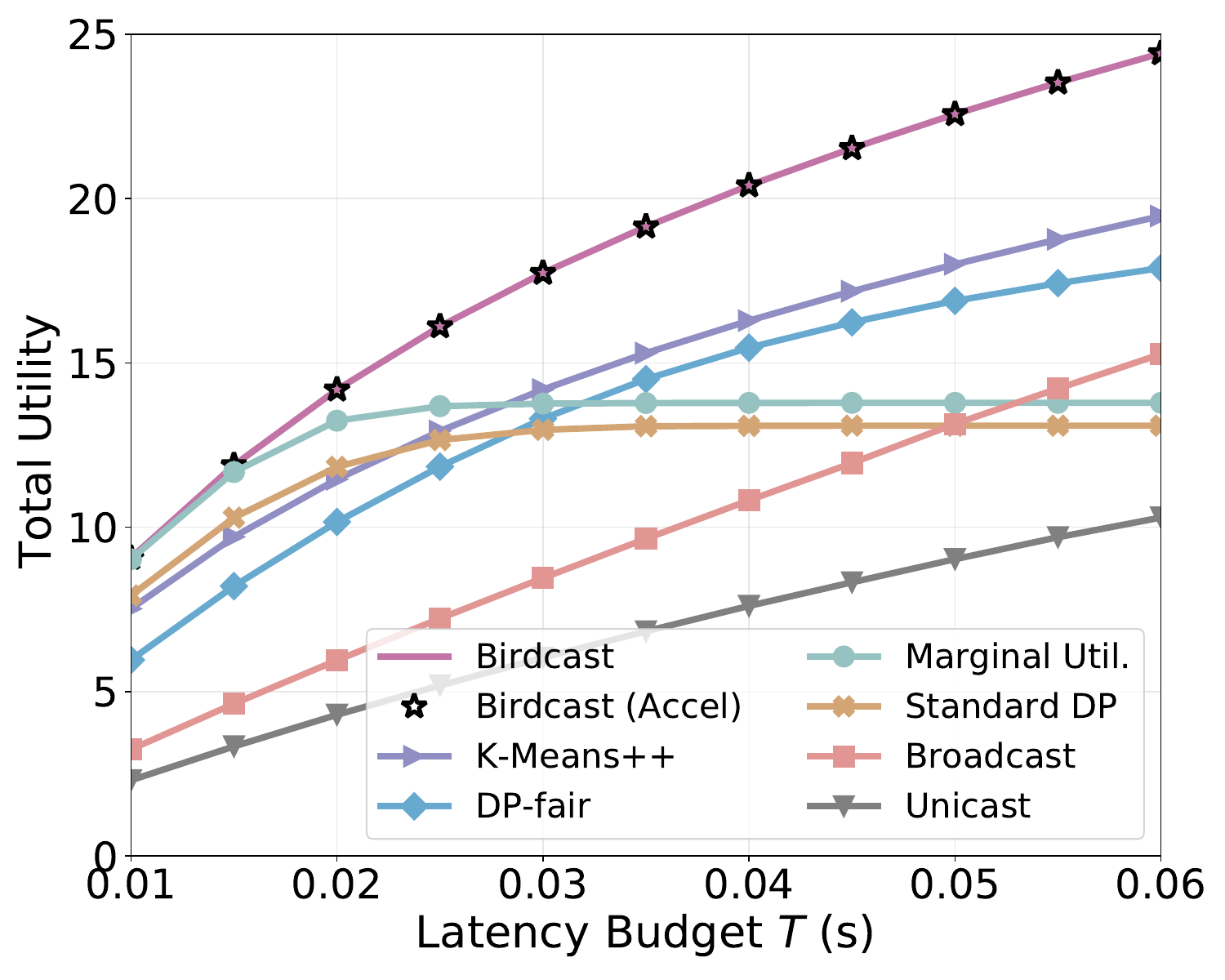}
  \caption{Total utility versus the latency budget $T$.}\label{fig:utility_latency}
\end{subfigure}
\hfill
\begin{subfigure}[b]{0.24\textwidth} 
  \includegraphics[width=\textwidth]{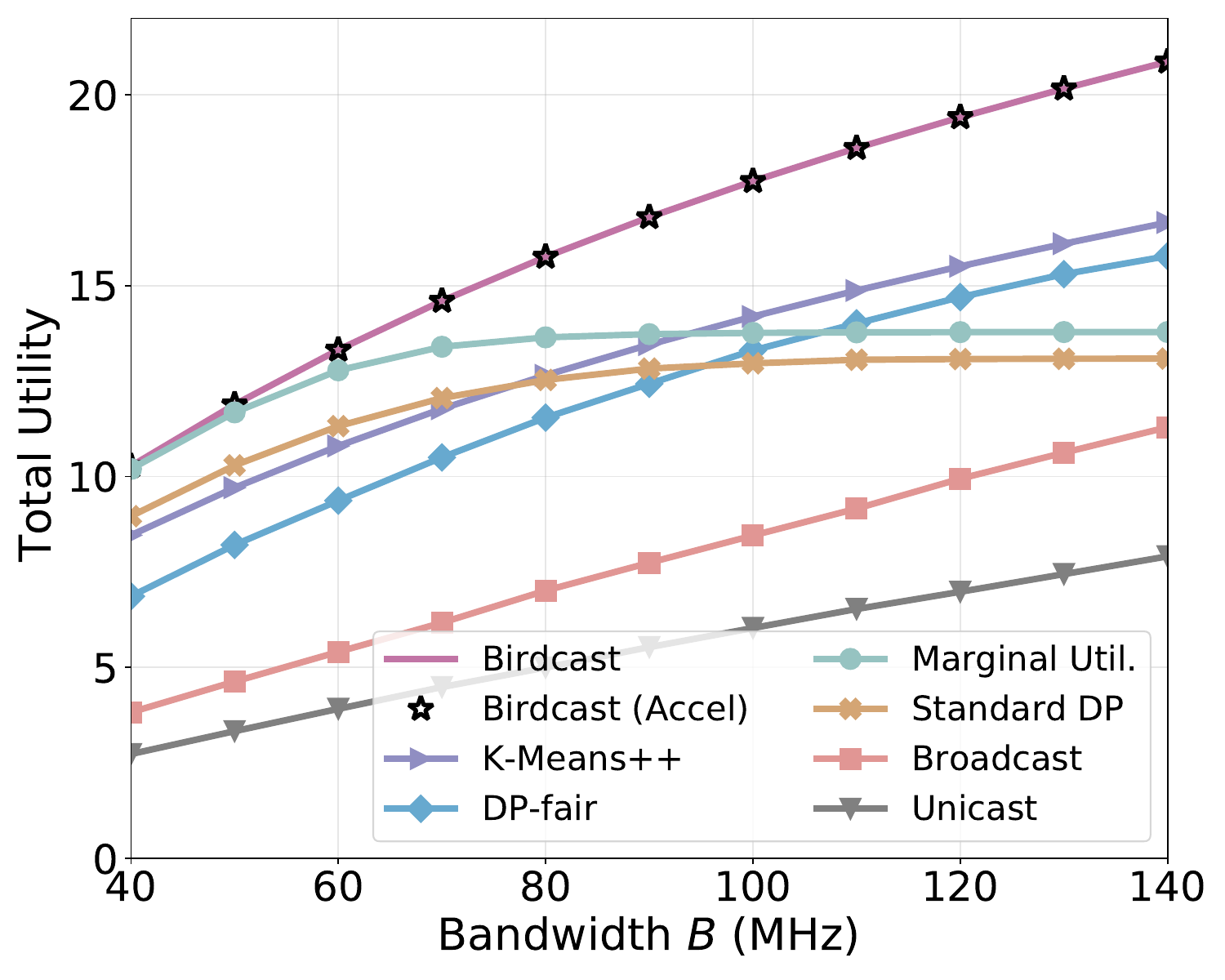}
  \caption{Total utility versus the bandwidth $B$.}\label{fig:utility_bandwidth}
\end{subfigure}
\\
\begin{subfigure}[b]{0.24\textwidth} 
  \includegraphics[width=\textwidth]{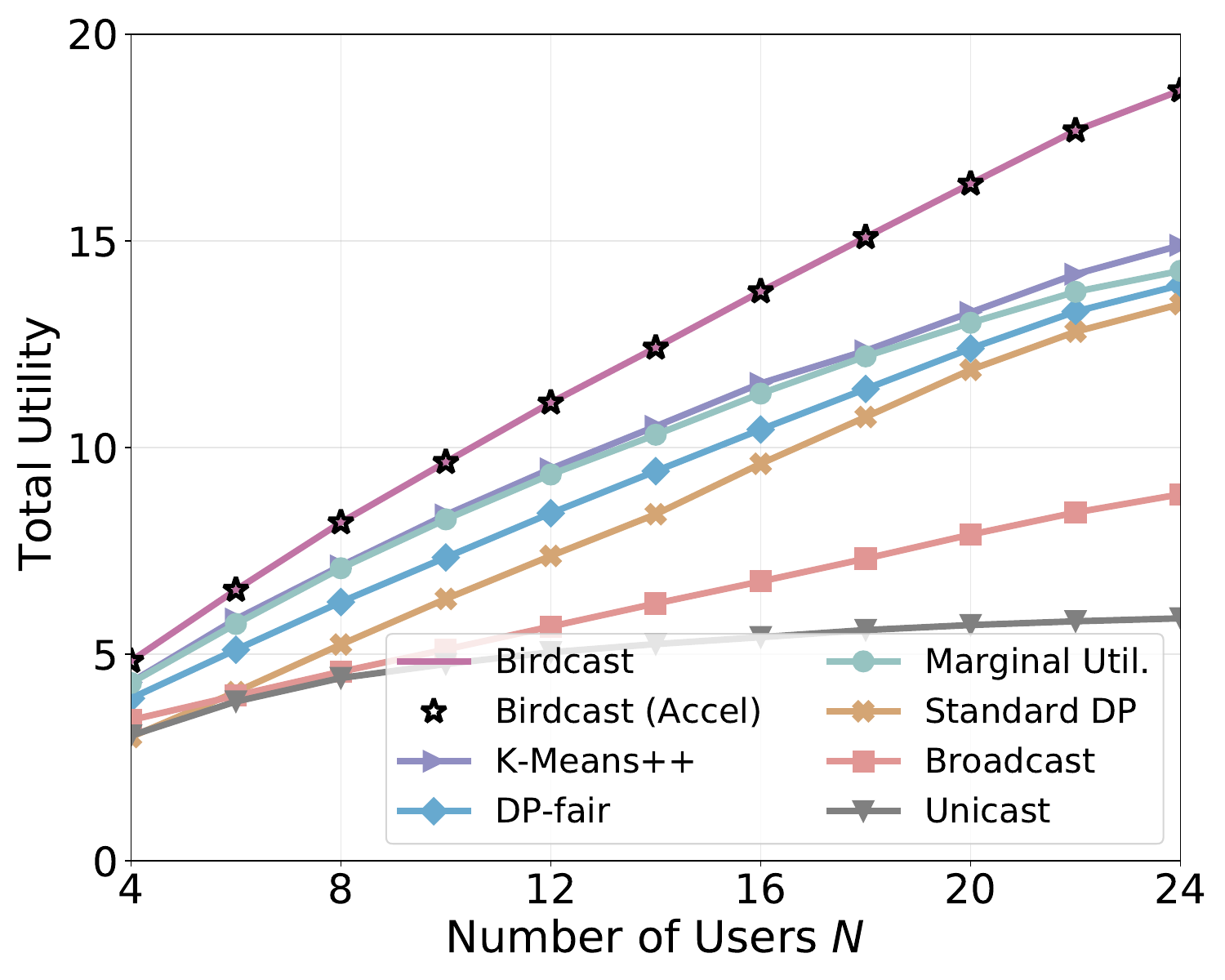}
  \caption{Total utility versus the number of users $N$.}\label{fig:utility_numveh}
\end{subfigure}
\hfill
\begin{subfigure}[b]{0.24\textwidth} 
  \includegraphics[width=\textwidth]{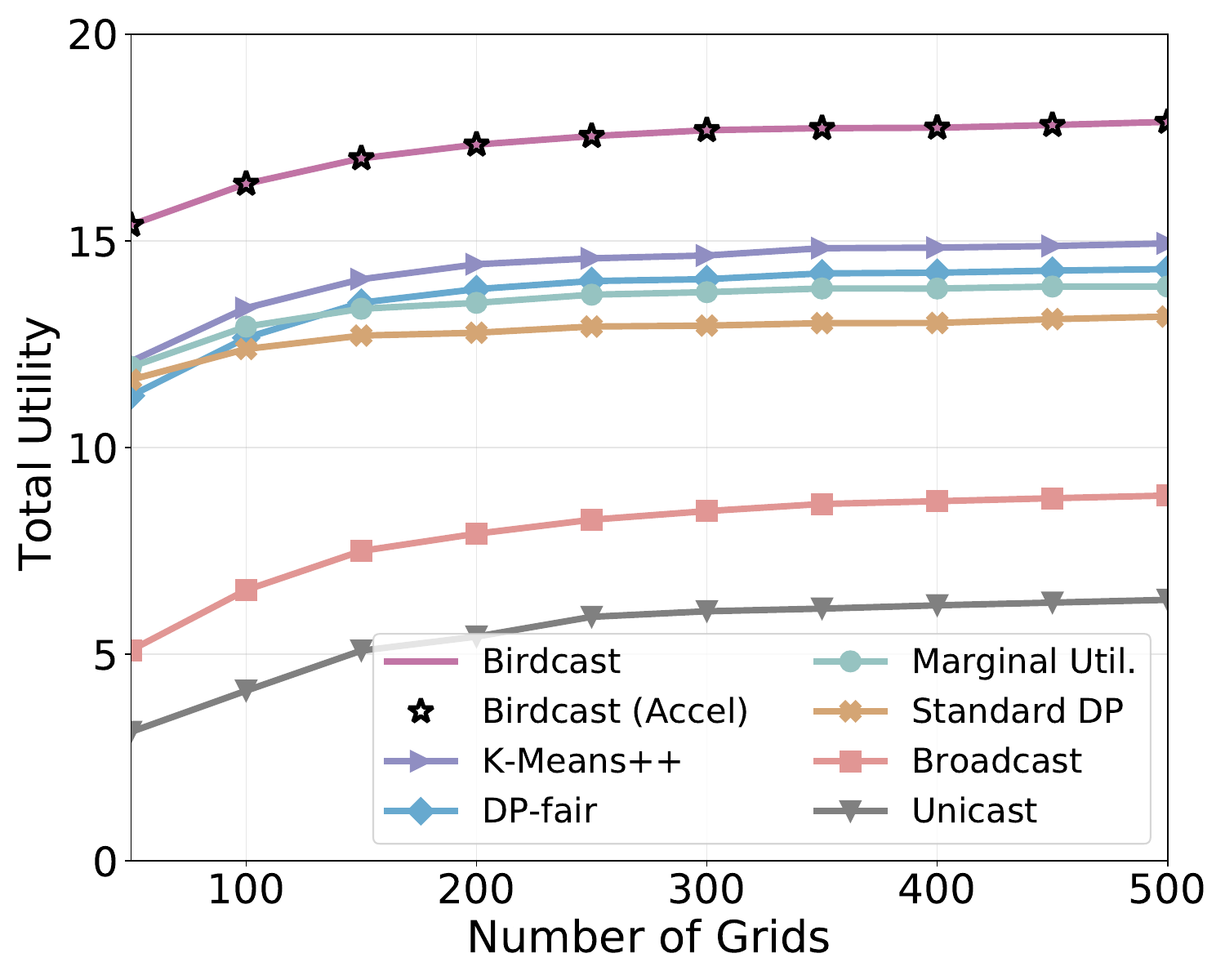}
  \caption{Total utility versus the number of grids $L$.}\label{fig:utility_numgrid}
\end{subfigure}
\caption{The total utility versus the network configurations. }\label{fig:utility}
\vspace{-0.0cm}
\end{figure}

\subsection{Utility Evaluation}
We first evaluate the system's total utility. Since standard CP datasets are specifically tailored for ego-user object detection, the number of users (vehicles) per scene is inherently limited (e.g., up to five in the V2X-Sim dataset). To rigorously assess utility performance and the effectiveness of our framework in dense environments, we simulate additional users based on the V2X-Sim dataset. Beyond the existing users with complete sensor data, we identify 17 to 22 supplementary users from the ground truth to serve as ``simulated'' users, yielding a total of 20 to 25 users per scene. For existing users, we generate the MoI directly from their sensor data. For the simulated users lacking raw sensor data, we synthesize the corresponding MoI based on their geometric positions and the density of surrounding vehicles.

In Fig. \ref{fig:utility_latency}, we present the total utility versus the latency budget $T$. As the latency budget $T$ relaxes, the total utility generally increases across all methods, as the system can accommodate more data transmission. 
Our Birdcast framework (employing the standard greedy algorithm) and its accelerated variant, Birdcast (Accel), consistently achieve the highest utility among the baselines, thanks to their adaptive grouping strategy. Unlike Broadcast, which is bottlenecked by the user with the worst channel condition, or Unicast, which suffers from resource exhaustion as it attempts to serve users sequentially, Birdcast efficiently balances spectral efficiency and user coverage. By dynamically clustering users, our Birdcast framework ensures that critical features are reliably delivered to the most relevant users within the latency constraint, yielding high utility. On the other hand, heuristic baselines such as K-means++ and DP-fair suffer from suboptimal group formations. Furthermore, Marginal Util. and Standard DP prematurely plateau after approximately 20 ms. This occurs because these heuristic algorithms commit to a single, high transmission rate for a specific feature grid without subsequently transmitting that same grid at a lower rate to reach a broader audience.

In Fig. \ref{fig:utility_bandwidth}, we show the utility versus the available bandwidth $B$. As anticipated, increasing the bandwidth yields a monotonic improvement in total utility across all schemes, as the relaxed resource constraint leads to higher data rates. Our Birdcast consistently yields the highest total utility, demonstrating superior spectral efficiency. In the low-bandwidth regime (e.g., $B = 40$ MHz), resource scarcity is the dominant bottleneck. In this setting, both Birdcast and Marginal Util. prioritize transmitting the BEV feature grids at the maximum data rate, thereby reaching only the user with the most favorable channel conditions. As a result, both methods exhibit comparable performance.
Nevertheless, as bandwidth increases, the performance gaps between Birdcast and others, including Marginal Util., increase due to their heuristic grouping logic.

Fig. \ref{fig:utility_numveh} illustrates the system's total utility as the number of users, $N$, increases from 4 to 24. As user density increases, the limited communication resources must be shared among a larger population, widening the performance gap between methods. Birdcast demonstrates superior scalability and robustness compared to baselines. By dynamically optimizing multicast grouping formation, Birdcast effectively balances the trade-off between multicasting gains and channel quality constraints. Heuristic baselines, such as Marginal Util and K-means++, fail to capture the complex interplay between user interest (MoI) and channel conditions, resulting in suboptimal groupings.

\begin{figure}[t]
\centering
\begin{subfigure}[b]{0.24\textwidth} \label{fig:time1}
  \includegraphics[width=\textwidth]{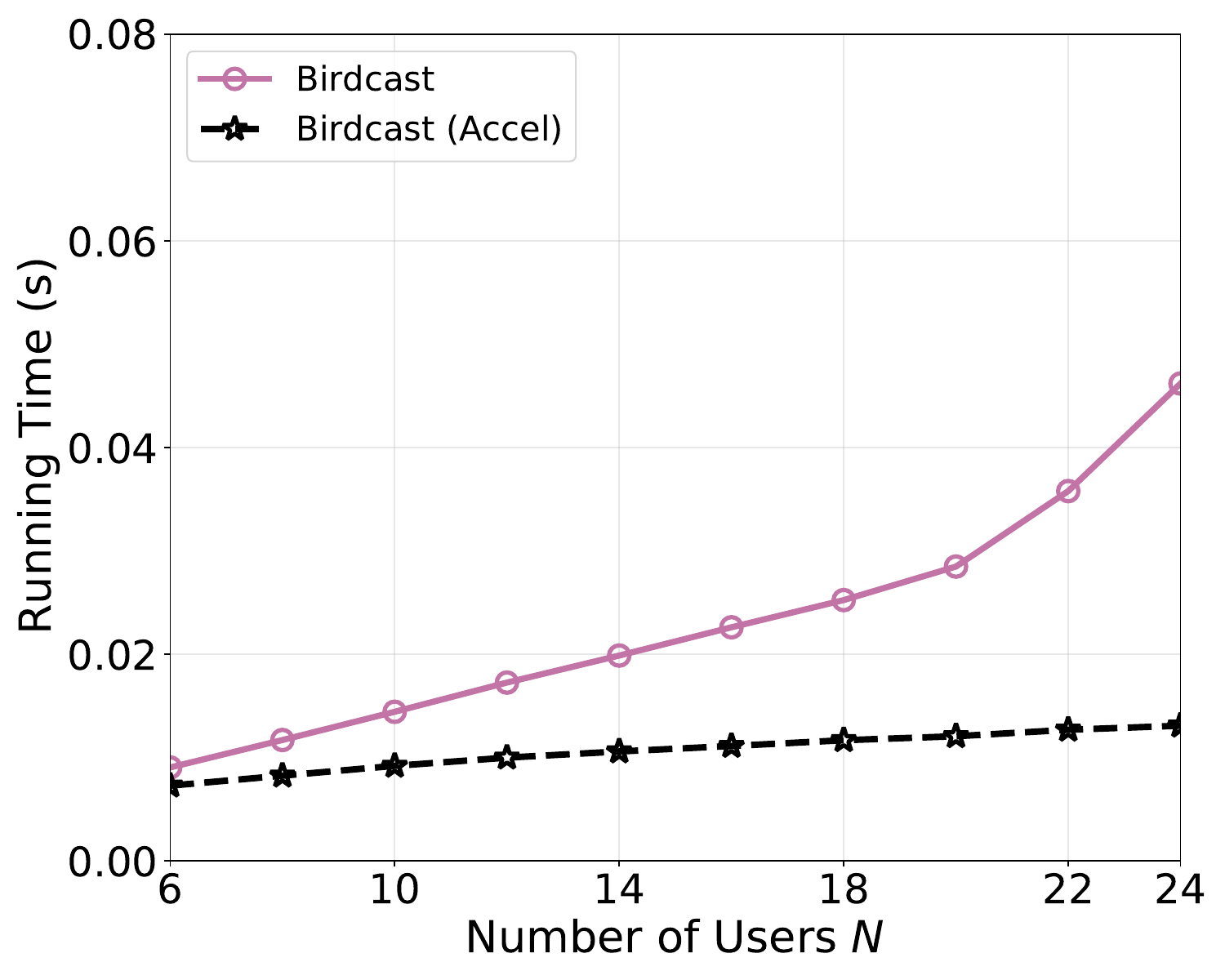}
  \caption{Running time versus the number of users $N$ with $L = 250$ grids.}
\end{subfigure}
\hfill
\begin{subfigure}[b]{0.24\textwidth} \label{fig:time2}
  \includegraphics[width=\textwidth]{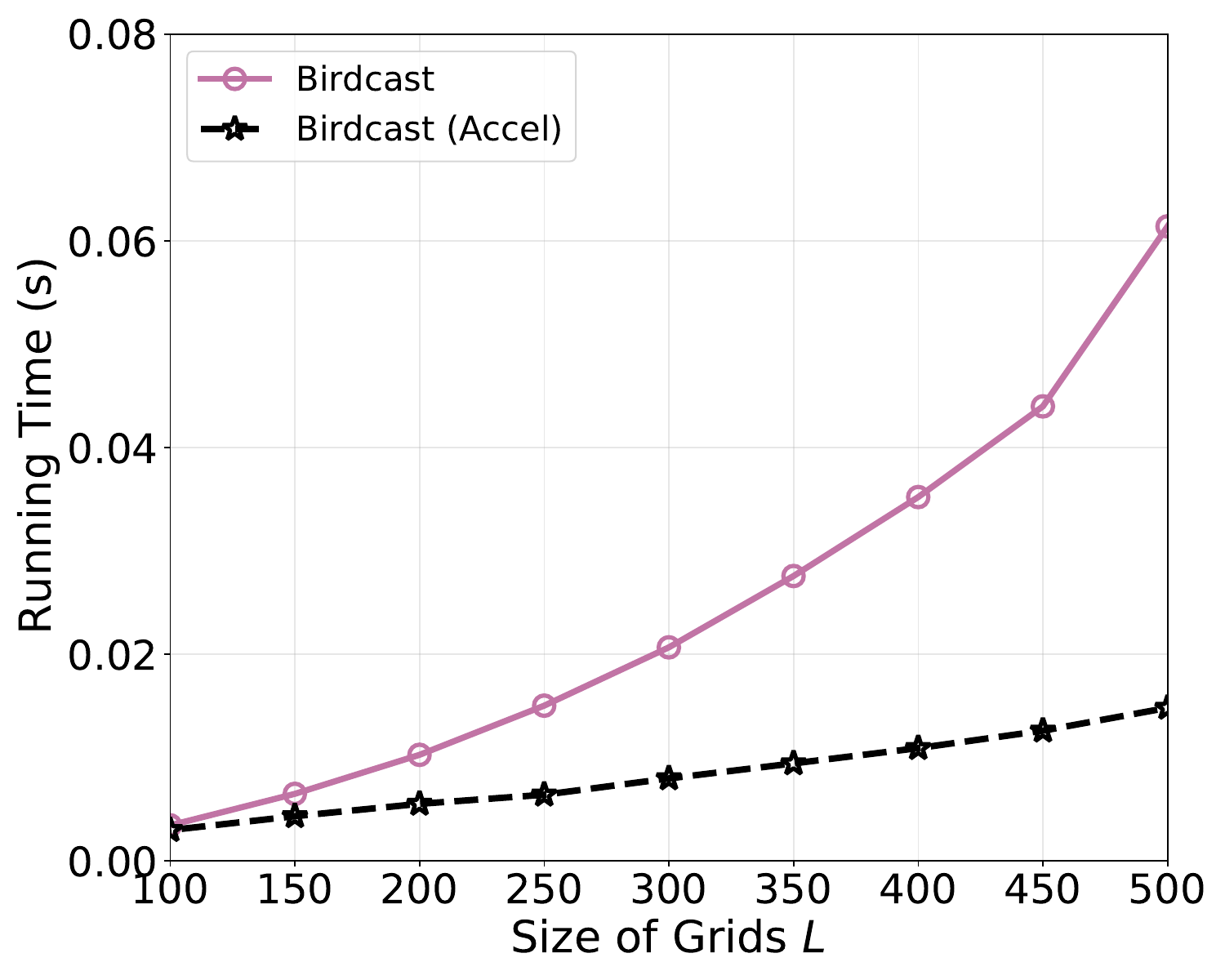}
  \caption{Running time versus the size of BEV grids $L$ with $N = 12$ users.}
\end{subfigure}
\caption{The running time versus network configurations.}\label{fig:time}
\vspace{-0.0cm}
\end{figure}

The total utility versus the number of BEV grids $L$ is presented in Fig. \ref{fig:utility_numgrid}. As the number of grids increases, the total utility of all schemes initially improves and then gradually converges to a plateau. This initial improvement can be expected: increasing the number of grids yields higher spatial resolution, allowing for more precise selection of grids. However, as $L$ increases further (e.g., beyond 250 grids), the marginal performance gains diminish, and the utility curves become nearly flat. This saturation occurs because the system capacity is fundamentally restricted by latency and bandwidth budgets rather than the granularity of grid selection. Consequently, partitioning the spatial region into an excessively large number of grids is unnecessary, as doing so increases computational complexity while yielding negligible utility benefits.



Finally, Fig. \ref{fig:time} illustrates the running time of our algorithms. Across different numbers of users and grids, both the greedy and accelerated greedy algorithms run efficiently. In particular, the accelerated greedy algorithm is substantially faster than the standard greedy algorithm and is very efficient at practical network scales (e.g., 13 ms for 24 users). Moreover, its running time grows approximately linearly with the number of users, demonstrating the scalability of our algorithms.



\begin{figure}[t]
\centering
\begin{subfigure}[b]{0.24\textwidth} 
  \includegraphics[width=\textwidth]{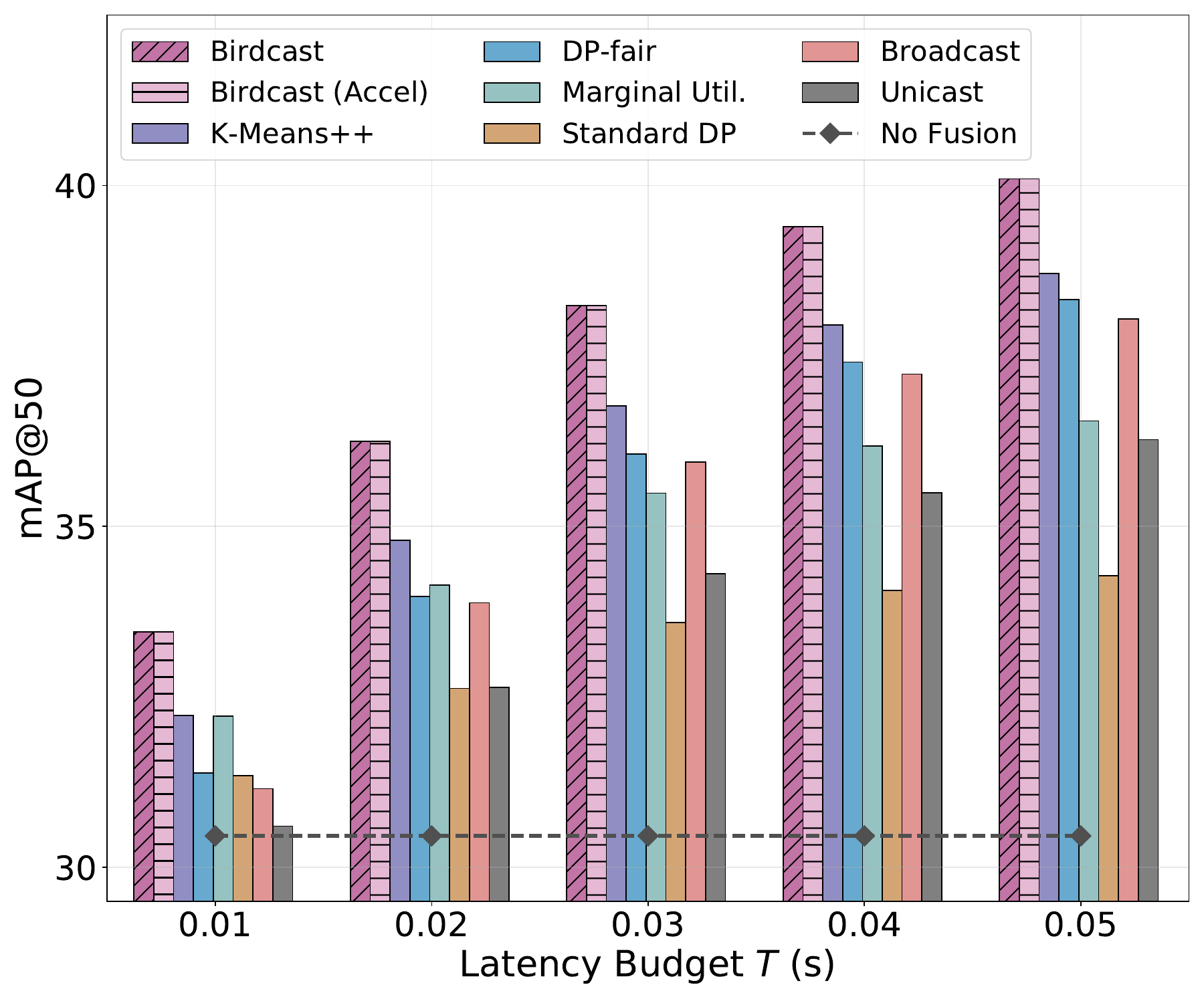}
  \caption{mAP@50 versus the latency budget $T$ with bandwidth $B = 100$ MHz.}\label{fig:map50_vs_latency}
\end{subfigure}
\hfill
\begin{subfigure}[b]{0.24\textwidth}
  \includegraphics[width=\textwidth]{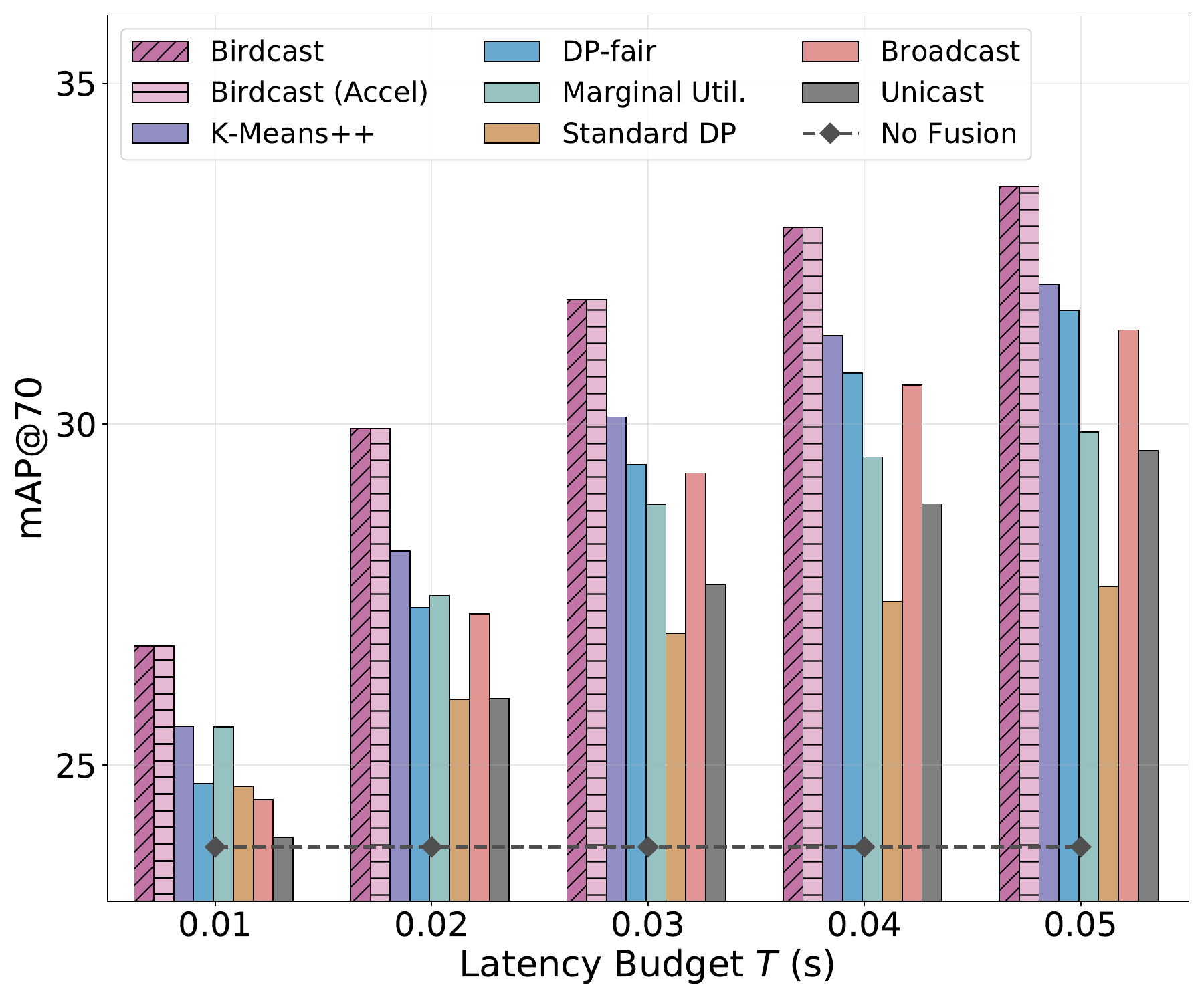}
  \caption{mAP@70 versus the latency budget $T$ with bandwidth $B = 100$ MHz.} \label{fig:map70_vs_latency}
\end{subfigure}
\caption{The perception performance (mAP@50/mAP@70) versus the latency budget $T$ with bandwidth $B = 100$ MHz.}\label{fig:map_vs_latency}
\vspace{-0.0cm}
\end{figure}

\begin{figure}[t]
\centering
\begin{subfigure}[b]{0.24\textwidth} 
  \includegraphics[width=\textwidth]{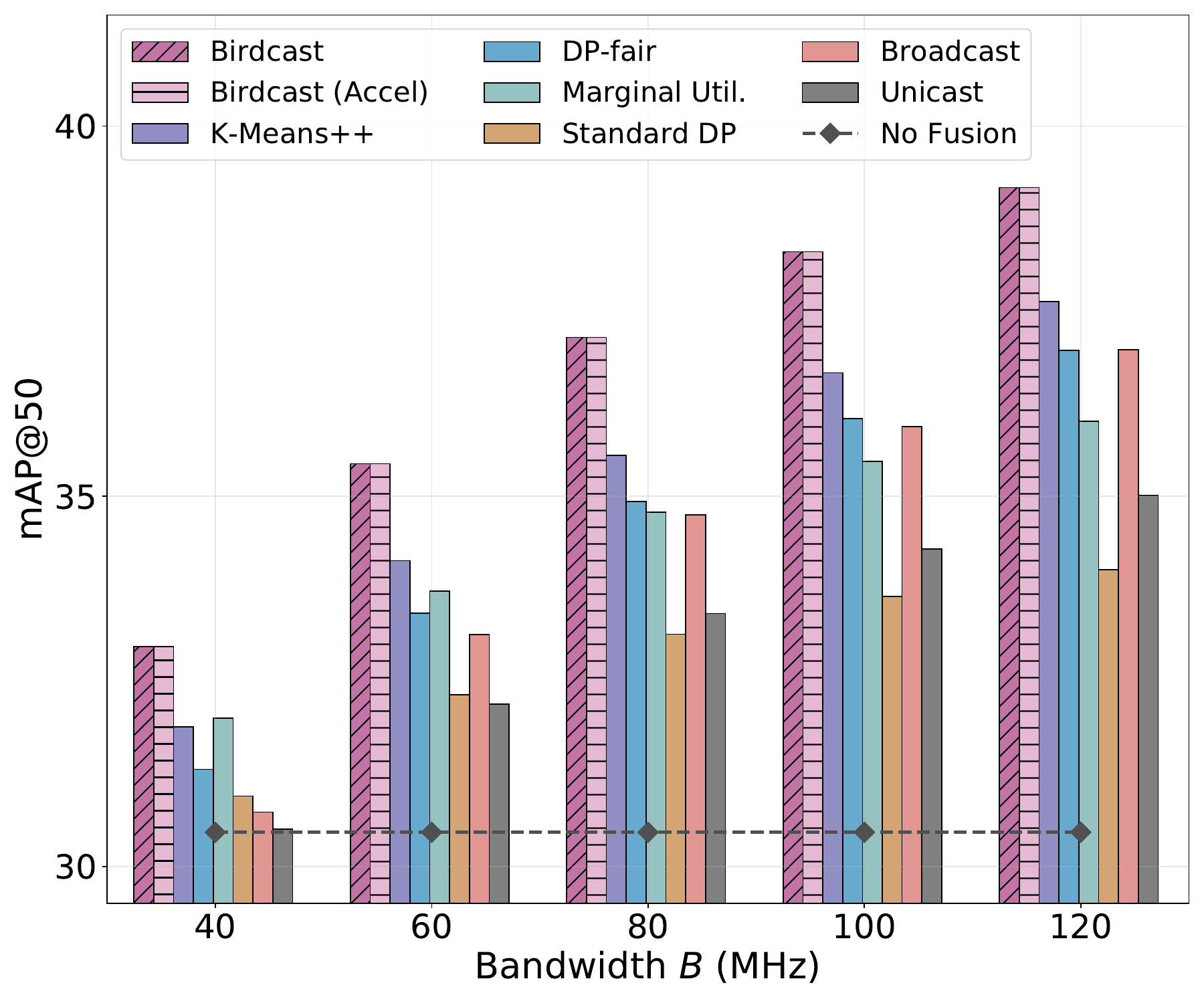}
  \caption{mAP@50 versus the bandwidth $B$ with latency budget $T=30$ ms.}\label{fig:map50_vs_bandwidth}
\end{subfigure}
\hfill
\begin{subfigure}[b]{0.24\textwidth} 
  \includegraphics[width=\textwidth]{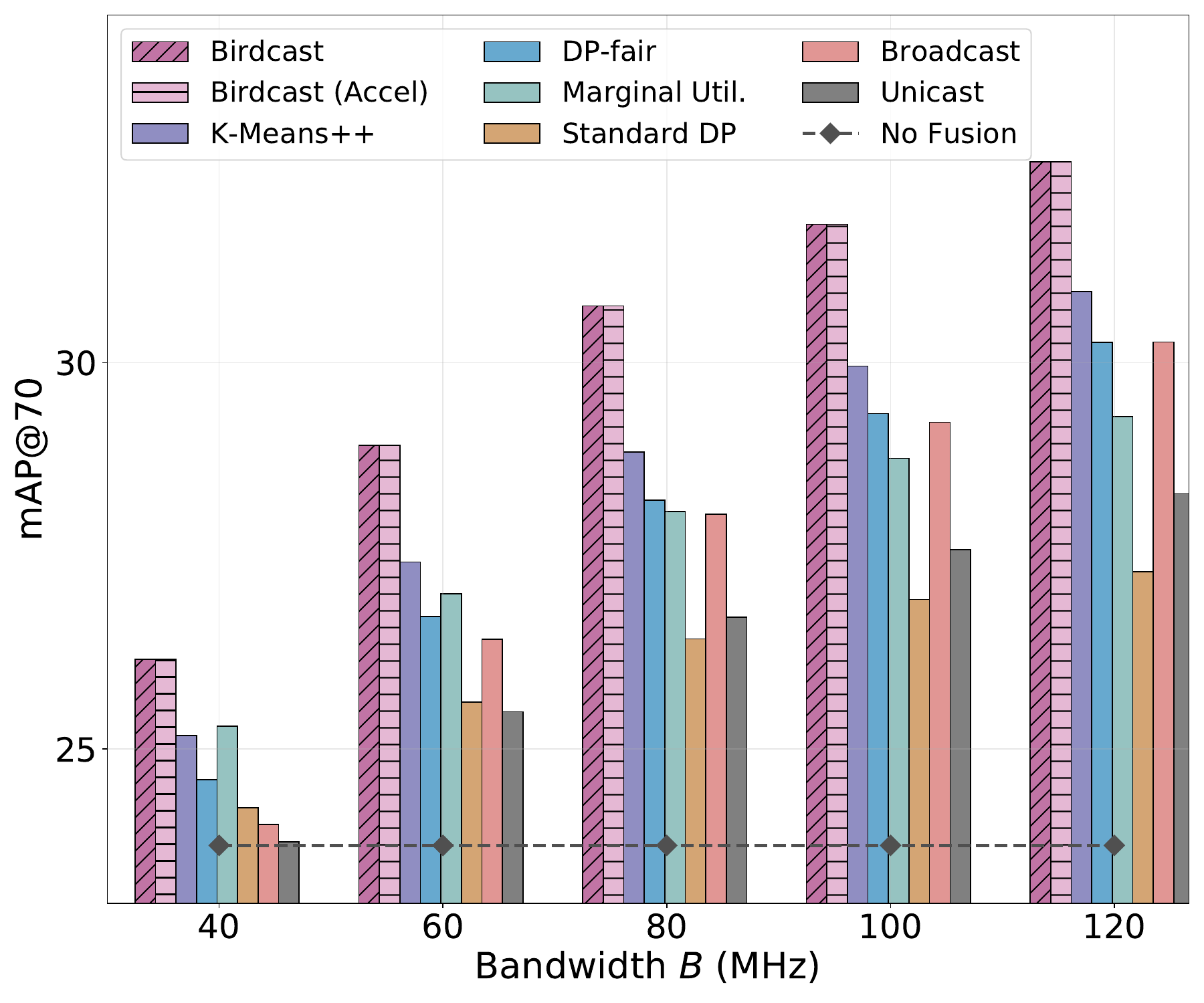}
  \caption{mAP@70 versus the bandwidth $B$ with latency budget $T=30$ ms.}\label{fig:map70_vs_bandwidth}
\end{subfigure}
\caption{The perception performance (mAP@50/mAP@70) versus the bandwidth $B$ with latency budget $T=30$ ms.}\label{fig:map_vs_bandwidth}
\vspace{-0.0cm}
\end{figure}

\begin{figure*}[t]
\centering
  \includegraphics[width=0.9\textwidth]{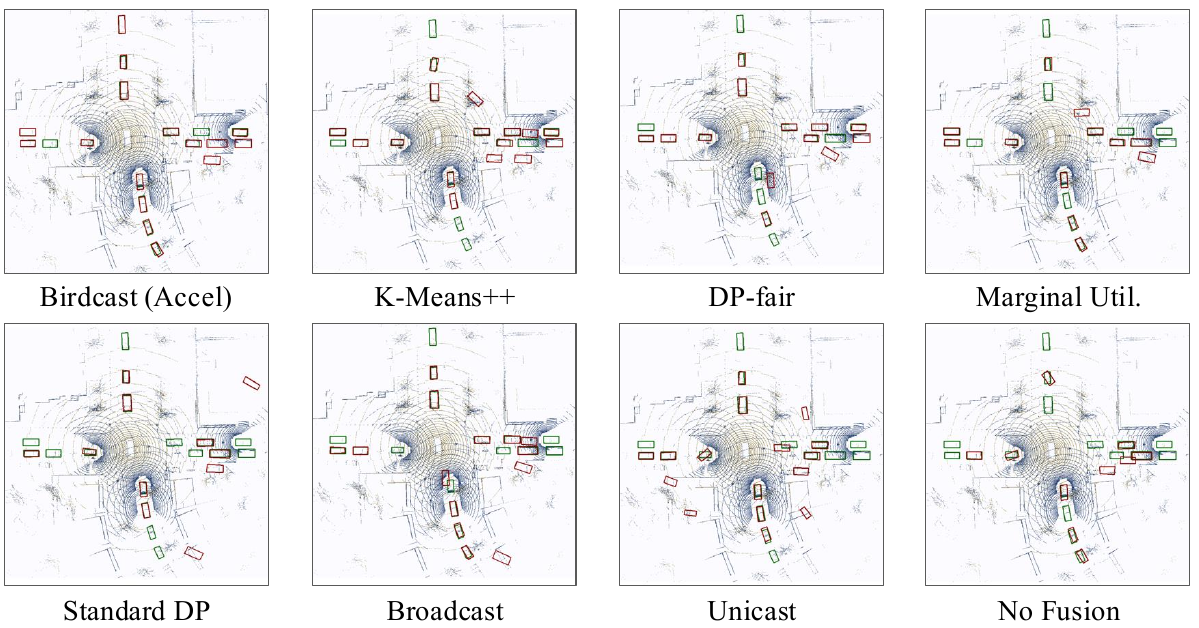}
  
\caption{Visualization of 3D detection results on V2X-Sim dataset. Red boxes are predictions while the green ones are ground truth. The bandwidth is $B=100$ MHz and the latency budget is $T=30$ ms.}\label{fig:visualization}
\vspace{-0.0cm}
\end{figure*}

\subsection{Perception Performance}

In this section, we present the overall perception performance of the Birdcast framework. To investigate the advantages of our scheme, we add another baseline: \textbf{No Fusion:} In this baseline, users perform perception tasks solely using their onboard sensors. No features are received or fused from the HVN, serving as the lower-bound benchmark for our system.

Fig. \ref{fig:map50_vs_latency} and Fig. \ref{fig:map70_vs_latency} show the perception performance, measured in mAP@50 and mAP@70, respectively, versus the latency budget $T$ with a fixed bandwidth $B$ of 100 MHz. The No Fusion baseline remains constant across all latency budgets since it relies entirely on local sensing, yielding the lowest mAP scores. For all other methods, the perception performance improves as the latency budget increases, allowing more informative features to be transmitted and fused. Birdcast and Birdcast (Accel) consistently outperform all baselines across the entire latency spectrum. Under extremely tight latency constraints (e.g., $T$ = 10 ms), Birdcast and Marginal Util. exhibit similar performance, because both methods default to transmitting grids at the peak rate to target the user with the best channel conditions in such restrictive settings. However, as the latency budget increases, Birdcast effectively capitalizes on the additional time to multicast features to judiciously formed groups, achieving significant mAP gains. In contrast, other heuristic approaches, including Marginal Util. and Standard DP, attain inferior performance.

Fig. \ref{fig:map50_vs_bandwidth} and Fig. \ref{fig:map70_vs_bandwidth} illustrate the perception performance as a function of the available system bandwidth $B$ under a fixed latency budget $T$ of 30 ms. Similar to the latency trends, an increase in bandwidth yields a monotonic improvement in mAP across all communication-based schemes, as relaxed resource constraints enable higher transmission rates. Birdcast and Birdcast (Accel) consistently achieve the highest mAP@50 and mAP@70 across all bandwidth regimes. 


Finally, we present the qualitative visualization results in Fig. \ref{fig:visualization}, offering an intuitive comparison of perception performance across all methods. In the No Fusion baseline, the user suffers from severe occlusion and a limited sensing range, leading to many missed detections. While conventional Unicast and Broadcast introduce external features, their low resource efficiency prevents the consistent delivery of critical spatial data, leaving some distant or heavily occluded objects undetected. Similarly, heuristic approaches, such as Marginal Util. and K-Means++, exhibit improved coverage but still produce imprecise bounding boxes due to suboptimal feature selection and user grouping. In contrast, Birdcast (Accel) judiciously prioritizes the most highly valued feature grids for properly formed multicast groups. As depicted in the visualization, Birdcast (Accel) successfully mitigates occlusion and extends the effective perception range, yielding highly accurate predictions and demonstrating its robust real-world applicability.





\section{Conclusion} \label{sec:conclusion}
In this paper, we have proposed \textit{Birdcast}, a multicasting framework for infrastructure-assisted collaborative perception (CP), featuring users with heterogeneous interests. To facilitate communication-efficient transmissions, the HVN and users collaboratively generate maps of interest (MoIs). Leveraging these MoIs, we have formulated a joint feature selection and multicast grouping problem aimed at maximizing network-wide utility while adhering to latency constraints. Given the NP-hard nature of this optimization problem, we have transformed it from a grid-rate perspective. By establishing the submodularity of the objective function, we have devised an accelerated greedy algorithm that is computationally efficient and guarantees an approximation ratio of $1 - 1/\sqrt{e}$. Extensive simulations on the V2X-Sim dataset demonstrate that Birdcast significantly outperforms baseline benchmarks, achieving up to a 27\% improvement in total utility and a 3.2\% gain in mean Average Precision for object detection.

While this study addresses interest-aware multicasting for infrastructure-aided CP, the Birdcast framework can be readily applied to diverse multicasting applications characterized by heterogeneous user interests. Furthermore, as the current system relies on downlink transmission from the infrastructure, future work can explore multicasting problems for vehicle-to-vehicle CP, enabling users to directly exchange features with each other under communication constraints.





\bibliographystyle{IEEEtran} 
\bibliography{reference}

\vfill

\end{document}